\documentclass[11pt,letterpaper]{article}
\usepackage[utf8]{inputenc}
\usepackage[margin=1in]{geometry}

\usepackage{amsmath}
\usepackage{amsfonts}
\usepackage{amssymb}
\usepackage{amsthm}
\usepackage[normalem]{ulem}
\usepackage[table,xcdraw,dvipsnames]{xcolor}
\definecolors{RawSienna}
\usepackage{hyperref}
\usepackage[capitalise,nameinlink]{cleveref}

\hypersetup{
    colorlinks=true,
    pdfpagemode=UseNone,
    citecolor=OliveGreen,
    linkcolor=NavyBlue,
    urlcolor=Magenta,
    pdfstartview=FitW
}

\usepackage{bbm}
\usepackage{braket}
\usepackage{graphicx}
\usepackage{algorithm}
\usepackage{algorithmic}

\makeatletter
\renewcommand\@fnsymbol[1]{
\ifcase#1\or$\dagger$\or$\ddagger$\or$\mathsection$\fi
}
\makeatother

\newtheorem{theorem}{Theorem}
\numberwithin{theorem}{section}
\newtheorem{lemma}[theorem]{Lemma}
\newtheorem{definition}[theorem]{Definition}
\newtheorem{proposition}[theorem]{Proposition}
\newtheorem{corollary}[theorem]{Corollary}

\Crefname{ALC@unique}{Line}{Lines}

\newcommand{\LS}{\left(}
\newcommand{\RS}{\right)}
\newcommand{\LM}{\left[}
\newcommand{\RM}{\right]}
\newcommand{\LL}{\left\{}
\newcommand{\RL}{\right\}}

\newcommand{\X}{{\mathcal X}}
\newcommand{\Y}{{\mathcal Y}}
\newcommand{\Dist}[1]{{\mathcal D_{#1}}}
\newcommand{\Real}{{\mathbb R}}
\newcommand{\Chi}[1]{{\chi_{#1}}}
\newcommand{\IP}[2]{{\langle{#1},{#2}\rangle}}
\newcommand{\Deg}{{\mathrm{deg}}}
\newcommand{\Bs}{{\mathrm{bs}}}
\newcommand{\Norm}[1]{{\left\lVert{#1}\right\rVert}}

\newcommand{\Protocol}{{\mathcal P}}
\newcommand{\Abs}[1]{{\left|{#1}\right|}}
\newcommand{\Tr}{{\mathrm{tr}}}
\newcommand{\Trans}{{\mathrm T}}
\newcommand{\Matrix}{{\mathcal M}}
\newcommand{\Indicator}{{\mathbbm1}}
\newcommand{\Argmax}{{\mathrm{argmax}}}
\newcommand{\hatDist}[1]{{\widehat{\mathcal D}_{#1}}}

\title{A Lifting Theorem for Hybrid Classical-Quantum \\ Communication Complexity}

\author{
Xudong Wu\thanks{State Key Laboratory of Novel Software Technology, Nanjing University, Nanjing 210023, China.}\and
Guangxu Yang\thanks{University of Southern California, Los Angeles, California 90089, USA.}\and
Penghui Yao\footnotemark[1]\hspace{1ex}\thanks{Hefei National Laboratory, Hefei 230088, China.}
}
\date{\today}

\begin{document}

\maketitle

\begin{abstract}
We investigates a model of hybrid classical-quantum communication complexity, in which two parties first exchange classical messages and subsequently communicate using quantum messages.
We study the trade-off between the classical and quantum communication for composed functions of the form $f\circ G^n$, where $f:\{0,1\}^n\to\{\pm1\}$ and $G$ is an inner product function of $\Theta(\log n)$ bits.
To prove the trade-off, we establish a novel lifting theorem for hybrid communication complexity.
This theorem unifies two previously separate lifting paradigms: the query-to-communication lifting framework for classical communication complexity and the approximate-degree-to-generalized-discrepancy lifting methods for quantum communication complexity.
Our hybrid lifting theorem therefore offers a new framework for proving lower bounds in hybrid classical-quantum communication models. 

As a corollary, we show that any hybrid protocol communicating $c$ classical bits followed by $q$ qubits to compute $f\circ G^n$ must satisfy $c+q^2=\Omega\big(\max\{\Deg(f),\Bs(f)\}\cdot\log n\big)$, where $\Deg(f)$ is the degree of $f$ and $\Bs(f)$ is the block sensitivity of $f$.
For read-once formula $f$, this yields an almost tight trade-off: either they have to exchange $\Theta\big(n\cdot\log n\big)$ classical bits or $\widetilde\Theta\big(\sqrt n\cdot\log n\big)$ qubits, showing that classical pre-processing cannot significantly reduce the quantum communication required.
To the best of our knowledge, this is the first non-trivial trade-off between classical and quantum communication in hybrid two-way communication complexity.
\end{abstract}

\section{Introduction}

{\em Hybrid quantum computation} delegates part of the computation to classical processors, integrating classical control, memory, and processing alongside quantum subroutines that execute essential quantum computations.
Since fully fault-tolerant quantum computers have not yet been realized, this model effectively captures the current NISQ (Noisy Intermediate-Scale Quantum) era~\cite{Pre18}, where quantum hardware remains limited and the classical resources complement and amplify the computational power.
A plethora of hybrid algorithms have been proposed, including the Variational Quantum Eigensolver (VQE)~\cite{PMSYZLAO14}, Quantum Approximate Optimization Algorithm (QAOA)~\cite{FGG14}.
Various hybrid quantum computation models have also been studied, demonstrating that quantum computers can achieve polynomial or even exponential advantages over classical computation, as seen the models such as $\textsf{DQC}_k$~\cite{KL98,MFF14}, $\textsf{NISQ}$~\cite{CCHL23}.

In a parallel line of research, quantum computing has also shown polynomial and even exponential advantages in communication complexity.
A substantial body of work has sought the weakest quantum communication model that still outperforms the strongest classical communication model~\cite{BCW98,Raz99,BJK08,GKKRW07,RK11,Gav20,GRT22,Gav19,Gav21,GGJL25,YZ25}.
Unlike time complexity, we have plenty of mathematical tools to prove lower bounds on communication complexity.
Consequently, the quantum advantage in communication complexity is unconditional, which does not rely on any unproven computational assumptions. 

In the NISQ era, researchers have proposed leveraging the exponential gap between classical and quantum communication to demonstrate quantum advantages.
For instance, Kumar, Kerenidis, and Diamanti experimentally implemented a quantum communication protocol for~\textsf{Hidden-matching} problem introduced by Bar-Yossef, Jayram, and Kerenidis~\cite{BJK08}.
Building on this line of work, Aaronson, Buhrman, and Kretschmer introduced the concept of {\em quantum information supremacy}~\cite{ABK24}: an experimental demonstration in which a quantum device solves a task using significantly fewer qubits than the number of bits required by any classical algorithm.
Their proposed task was again based on the classical–quantum exponential separation exhibited in the~\textsf{Hidden-matching} problem.
Most recently, quantum information supremacy has been experimentally demonstrated using a trapped-ion quantum computer~\cite{KGDGGGHMNHA25}.

In this work, we investigate the power of hybrid classical-quantum communication through the lens of communication complexity, aiming to characterize the advantages offered by combining classical and quantum communication resources.
We adopt the standard model of communication complexity~\cite{Yao79} and that of quantum communication complexity~\cite{Yao93}, both originally introduced by Yao.
In the hybrid classical–quantum communication complexity model, a protocol proceeds in two stages.
In the first stage, the parties exchange classical messages and perform local classical computations.
In the second stage, they perform local quantum operations and exchange quantum messages.
Our focus is on understanding the trade-off between classical and quantum communication costs within this hybrid framework.

\subsection{Our results}

In this paper, we study the trade-off between classical and quantum communication complexities within the hybrid classical–quantum communication model.
A central question we address is whether it is possible to simultaneously reduce both classical and quantum communication costs compared to purely classical or purely quantum protocols.
To this end, we investigate the hybrid classical-quantum communication complexity of the function family $f\circ G^n$, where $f$ takes an $n$-bit input, and $G$ is an inner product function on $\Theta(\log n)$ bits.
We establish a hybrid query-to-communication lifting theorem.

\begin{theorem}[Informal]
If $f\circ G^n$ can be computed by first communicating $c$ classical bits, followed by communicating $q$ qubits, then there is an $O\big(c/b)$-depth deterministic decision tree for the outer function $f$ such that $f$ restricted to any outcome of the query has approximate degree $O\big(q/b\big)$.
\end{theorem}

As an application, we prove a lower bound on the classical-quantum trade-off.

\begin{theorem}[Informal]
\label{thm:informal_tradeoff}
If $f\circ G^n$ can be computed by first communicating $c$ bits deterministically, followed by communicating $q$ qubits, then $c+q^2=\Omega\big(\max\{\Deg(f),\Bs(f)\}\cdot\log n\big)$.
In particular, for a read-once formula $f$, we have a tight trade-off: either $c=\Omega\big(n\log n\big)$ or $q=\Omega\big(\sqrt n\log n\big)$.
\end{theorem}

Our results refute the possibility that classical pre-processing can reduce the number of subsequent quantum bits.
This resembles the result about the trade-off between classical and quantum memory in memory-sample lower bounds for learning studied by Liu, Raz, and Zhan \cite{LRZ23}.

To the best of our knowledge, this is the first non-trivial trade-off between classical and quantum communication in a hybrid two-way communication model.
The research about the lower bounds on classical and quantum communication complexity has a long history \cite{LS09}, while the underlying techniques differ substantially.
Classical lower bounds are typically established via combinatorial methods \cite{KN97}, whereas quantum lower bounds often rely on analytic techniques \cite{LS09}.
Moreover, the query-to-communication lifting theorems for the two models were also developed independently.
In this paper, we unify these approaches through a novel lifting mechanism, which we believe offers new insights into establishing lower bounds in hybrid communication complexity.

\bigskip

Our main technical tool is the following theorem, which lifts the approximate degree to quantum communication complexity for composed functions restricted to rectangles.

\begin{theorem}[Informal]
Let $R=U\times V$ be a rectangle in the input domain of $f\circ G^n$.
If the uniform random variables on $U$ and $V$ are both $0.99$-dense, the quantum communication cost of $f\circ G^n$ restricted to $R$ is $\Omega\big(\Deg_\varepsilon(f)\cdot\log n\big)$.
\end{theorem}

\subsection{Related works}

\paragraph{Hybrid quantum computation.} Hybrid quantum communication complexity has been explored in several distinct settings.
In 2008, Gavinsky, Regev, and de Wolf studied the hybrid simultaneous message passing (SMP) model, in which one party sends a quantum message to a referee, the other sends a classical message, and the referee computes the function.
They established an almost tight bound on the quantum-classical communication complexity of \textsf{EQUALITY} in this model~\cite{GRW08}.
More recently, Arunachalam, Girish, and Lifshitz~\cite{AGL23} investigated the one-clean-qubit model of quantum communication, inspired by the quantum circuit complexity class $\textsf{DQC}_1$, where one qubit is in a pure state and all other qubits are maximally mixed.
They presented an explicit example demonstrating an exponential separation between the one-clean-qubit model and classical communication.
Lin, Wei, and Yao~\cite{LWY22} examined the hybrid classical-quantum communication complexity of generating classical correlations between two players -- a simpler task than computing a function -- and developed several lower bound techniques based on variants of nonnegative ranks and positive semidefinite (PSD) ranks.
Despite these advances, the field of hybrid classical–quantum communication complexity remains largely unexplored.
In contrast to the well-developed classical and quantum communication frameworks, there currently exist few general techniques for proving lower bounds in the hybrid setting.

The power of hybrid quantum computation has also been studied across a variety of computational models. One particularly relevant line of research concerns hybrid query complexity, a close analog of communication complexity.
In the query model, quantum algorithms are allowed a limited number of quantum queries, often interleaved with classical queries.
Researchers have developed several powerful techniques to prove lower bounds on hybrid query complexity.
Regev and Schiff initiated the study of Grover's search with a faulty oracle~\cite{RS08}, proving that no quantum speedup is possible when each oracle query may fail with small probability.
Built on this work, Rosmanis~\cite{Ros24} derived a tight bound for preimage search in this setting.
Hamoudi, Liu, and Sinha~\cite{HLS24} later extend the compressed-oracle framework introduced by Zhandry~\cite{Zha19} to establish tight bounds for collision finding in a hybrid query setting.
A related line of work~\cite{CCL23,CM20,AGS22,HL22,CH22} have proved lower bounds for hybrid algorithms in the so-called $d$-QC model, where $d$ quantum queries are interleaved with a polynomial number of classical queries.
Sun and Zheng~\cite{SZ19} studied decision trees in which each node corresponds to a quantum circuit that makes at most $q$ quantum queries and ends with a measurement.
They proved a quantum query complexity lower bound of $\Omega\big(\Bs(f)/q+\sqrt{\Bs(f)})$ for any function $f$.

In the context of quantum circuit models, several hybrid computational paradigms have been proposed and studied.
Knill and Laflamme~\cite{KL98}, followed by subsequent works~\cite{MFF14}, introduced the $\textsf{DQC}_k$ model, where the quantum circuits have access to $k$ clean qubits while the remaining qubits are maximally mixed, to capture the NMR approach to quantum computing.
The computational power of $\textsf{DQC}_k$ has been extensively studied by a series of works~\cite{ABKM17,JM24,Gir25}.
Chen, Cotler, Huang, and Li have proposed the complexity class \textsf{NISQ}~\cite{CCHL23} consisting of all problems solvable by a polynomial-time probabilistic classical algorithm equipped with access to a noisy quantum device.

Recently, Liu, Raz, and Zhan~\cite{LRZ23} initiated the study of learning with classical-quantum hybrid memory. They have established a tight trade-off among classical memory, quantum memory, and sample complexity for several learning tasks, refuting the possibility that a small amount of quantum memory significantly reduces the size of classical memory required for efficient learning on these problems.

\paragraph{Query-to-communication lifting.} Query-to-communication lifting theorems are generic methods for translating query complexity lower bounds to communication complexity lower bounds using a suitable base function composed with a gadget.
In classical two-party communication complexity, query-to-communication lifting theorems are known~\cite{RM97,GLMWZ16,GPW20,CFKMP19,CFKMP21,LMMPZ22,MYZ25} with sufficiently large gadgets.
These results have yielded diverse applications in various areas, including monotone circuit complexity, proof complexity, combinatorial optimization, and others.
In contrast, for quantum complexity, obtaining a general lifting theorem that translates the lower bounds of the quantum query into the lower bounds of the quantum communication remains a major open problem~\cite{ABGJKL17,CCMP20}.
A related line of work lifts lower bounds from approximate degree to quantum communication complexity~\cite{She11,SZ09,LZ10}.
Notably, these two lines of work rely on fundamentally different techniques, so insights from one do not directly transfer to the other. 
Motivated by recent research on hybrid quantum computation, we introduce a new hybrid communication model. To our knowledge, neither of the existing techniques for proving communication lower bounds in the classical or the quantum settings can be directly applied to this hybrid model. 

\subsection{Proof overview}

To prove the trade-off between the classical and quantum communication complexity in our hybrid model, we build upon and unify two independent lines of work on lifting theorems.
The first line of work lifts deterministic/randomized query complexity to communication complexity using gadgets of size $\Theta(\log n)$, such as index functions \cite{GPW20}, inner product functions \cite{CFKMP19}, and low discrepancy functions \cite{CFKMP21}.
The second line concerns lifting approximate degree to quantum communication complexity via constant-size gadgets, including index-like functions \cite{She11}, inner product functions \cite{SZ09}, and strongly balanced functions \cite{LZ10}.
We combine these two techniques to analyze both the classical and quantum phases in a hybrid communication protocol. 

More specifically, consider a protocol that first exchanges $c$ bits in the first phase.
It partitions the input domain into $2^c$ disjoint rectangles.
Suppose that thereafter the parties proceed via $q$ qubits of quantum communication and compute a function $F$.
Restricting to any one of the $2^c$ rectangles reduces the problem to a pure quantum communication problem of cost at most $q$.
To lower bound the quantum communication complexity, we apply the generalized discrepancy method: it is known that for any real matrix $\Psi$ supported on rectangle $R$, we have a lower bound on the quantum communication complexity of $F$ restricted to $R$
$$\Omega\left(\log\frac{\IP F\Psi-0.1\Norm\Psi_1}{\Norm\Psi\sqrt{|R|}}\right).$$

We follow arguments of the approximate-degree to quantum communication complexity lifting framework developed by Sherstov \cite{She11}.
However, complication and challenge arises, because, in our hybrid model, the rectangle $R$ is inherited from the classical communication in the first phase and may be arbitrary, whereas prior proofs rely on more structure of the input domain.

Our main technical contribution is to introduce a generic condition on rectangles $R$, which we call {\em density}, under which we can still prove quantum communication complexity lower bounds via the approximate degree to quantum communication complexity lifting approach.
We believe that this contribution is of independent interest and anticipates further applications.
Surprisingly, this ``density" property turns out to coincide exactly with the property of rectangles considered in the prior query-to-communication lifting research \cite{GPW20,CFKMP19,CFKMP21}.
Concretely, let $G$ be an inner product gadget of size $\Theta(\log n)$.
For a composed function $F=f\circ G^n$, we select $\Psi=\psi\circ G^n$, where $\psi$ is a dual polynomial as suggested in \cite{She11}.
On a dense rectangle $R$, the values of $G^n$ over $R$ are nearly uniform.
Moreover, once we obtain a dense rectangle $R$, we could prove a tight lower bound via the generalized discrepancy with more sophisticated analysis.
Our analysis relies crucially on the gadget size being $\Theta(\log n)$ and on the linearity of the inner product gadget.

To argue that there exists a dense rectangle $R$ after the classical phase, we follow the framework of query-to-communication lifting to locate a dense rectangle among the $2^c$ rectangles, where $c$ is the classical communication cost.
The general idea is as follows.
We start with $R$ being the entire input domain.
With each transmitted bit, $R$ is partitioned into two parts, and the larger one is selected.
Whenever $R$ is not dense, we can always fix a set of coordinates $I\subseteq[n]$ such that: by choosing $z\in\{0,1\}^I$ arbitrarily, a rectangle $R'\subseteq R$ can be found, which satisfies that $G^I$ is constant on $R'$ (taking the value $z$), and $R'$ is dense with respect to unfixed coordinates.
Then we replace $R$ by $R'$ to restore density.
For any $f$ and $I\subseteq[n]$, the coordinates in $I$ can be fixed such that the degree of $f$ decreases by at most $|I|$.
After transmitting $c$ bits, the number of fixed coordinates is at most $O(c/\log n)$.
And the outer function $f$ is equivalent to a function of degree $\Deg(f)-O(c/\log n)$ (and of approximate degree $\sqrt{\Deg(f)-O(c/\log n)}$) when restricted to $R$.
Our lifting theorem for hybrid classical-quantum communication is obtained by applying approximate degree lifting to dense rectangles.

\subsection{Discussion and open problems}

Given the advent of the NISQ era, hybrid quantum computation has attracted growing attention in recent years.
It is therefore tempting to understand the computational power of hybrid quantum computation across different computational models, from both theoretical and experimental perspectives. 

\bigskip

To the best of our knowledge, the trade-off between classical and quantum communication complexity for read-once formulas $f$ is the first (non-trivial) tight trade-off between classical and quantum communication in hybrid two-way communication complexity.
Our results refute the possibility that classical pre-processing can substantially reduce the quantum communication required for the function $f\circ G^n$.

We employ the query-to-communication lifting mechanism for both classical and quantum communication complexity—versatile techniques that have been developed over the past decades.
Our results give rise to several interesting open problems for future research.

\begin{enumerate}
\item Our result is obtained by combining the query-to-communication lifting theorem~\cite{GLMWZ16,GPW20,CFKMP19,CFKMP21} for classical communication complexity and the approximate-degree-to-generalized-discrepancy lifting theorem~\cite{She11,SZ09,LZ10}.
Is it possible to improve our results to obtain a better trade-off: if $c\ll D(f)\cdot\log n$, then $q=\Omega\big(\Deg_\varepsilon(f)\cdot\log n\big)$?
Is it possible to prove a general query-to-communication lifting theorem for hybrid classical-quantum communication complexity?

\item Can trade-off analogous to \Cref{thm:informal_tradeoff} be established in the hybrid randomized-quantum communication model?
The statement $c+q^2=\Omega\big(\Bs(f)\cdot\log n\big)$ may hold by considering the hard distribution derived from the generalized discrepancy bound, while the statement for degree is unlikely to hold since the degree can be significantly larger
than the randomized query complexity.
Is is possible to further prove that $q=\Omega\big(\Deg_\varepsilon(f)\cdot\log n\big)$ when $c\ll R(f)\cdot\log n$ by combining query-to-communication lifting for \textsf{BPP} and approximate-degree-to-generalized-discrepancy lifting?

\item Lifting theorems for both classical and quantum communication complexity have been established for a variety of gadget functions.
However, our proofs critically rely on the linearity property of the inner-product gadget.
A natural question is therefore whether our results can be generalized to other gadget functions.

\item This work studies the hybrid classical-quantum communication model, in which players exchange classical messages and quantum messages.
What about the quantum-classical communication model, where quantum communication comes first?
Is it possible to prove a trade-off for this model, or more generally, for any model that consists of a constant number of alternating purely classical phases and purely quantum phases?

\item Proving a query-to-communication lifting theorem with a constant-sized gadget for classical communication complexity is a major open problem in communication complexity. Nevertheless, 
is it possible to prove a lifting theorem for a certain class of outer functions $f$?
For instance, can we prove a tight trade-off for $f=\mathsf{OR}$?
It would imply a trade-off for \textsf{SET DISJOINTNESS}, a central problem in communication complexity~\cite{KS92,Raz92,BJKS04,Raz03,She14}, whose trade-off between classical and quantum communication in hybrid communication complexity is widely open.

% \item Does there exist any communication problem that has an advantage in a hybrid classical-quantum (or quantum-classical) communication model?
% That is, a hybrid communication protocol for this problem transmits significantly fewer classical bits than a purely classical communication protocol, and fewer quantum bits than a purely quantum communication protocol.
% To the best of our knowledge, this problem is also open in the setting of hybrid query complexity.
\end{enumerate}

\subsection*{Acknowledgements}

Xudong Wu and Penghui Yao were supported by National Natural Science Foundation of China (Grant No. 62332009, 12347104), Quantum Science and Technology-National Science and Technology Major Project (Grant No. 2021ZD0302901), NSFC/RGC Joint Research Scheme (Grant No. 12461160276), Fundamental and Interdisciplinary Disciplines Breakthrough Plan of the Ministry of Education of China (No. JYB2025XDXM118), Natural Science Foundation of Jiangsu Province (No. BK20243060).
\section{Preliminaries}

\paragraph{Notations.} For a random variable $X\in\X$, denote its distribution by $\Dist X$.
So $\Dist X(x)=\Pr[X=x]$ for $x\in\X$.
And for an event $E$, the random variable $X\mid E$ follows the distribution $\Dist X(\cdot\mid E)$.
For a set $U$, we write $X\sim U$ to denote that the random variable $X$ is uniformly distributed over $U$.

Let $n$ be a positive integer and $J\subseteq[n]=\{1,\cdots,n\}$ be a set of coordinates.
We define $\bar J=[n]\setminus J$.
Let $f:\{0,1\}^J\to\{\pm1\}$ be a Boolean function.
For any $K\subseteq J$ and $z\in\{0,1\}^{J\setminus K}$, define $f_{K,z}:\{0,1\}^K\to\{\pm1\}$ such that
$$f_{K,z}(x)=f(x,z),\qquad\forall\ x\in\{0,1\}^K.$$
That is, the function $f_{K,z}$ is obtained by fixing the coordinates outside $K$ to be $z$.

\paragraph{Basic Fourier analysis.} For a Boolean function $f:\{0,1\}^n\to\Real$, the Fourier expansion gives
$$f=\sum_{S\subseteq[n]}\widehat f_S\Chi S,$$
where those $\Chi S:\{0,1\}^n\to\{\pm1\}$ satisfying $\Chi S(x)=\prod_{i\in S}(-1)^{x_i}$ are orthogonal with respect to the inner product $\langle f, g\rangle=\sum_xf(x)g(x)$. $\widehat f_S$ are Fourier coefficients of $f$ satisfying $\widehat f_S=2^{-n}\IP f{\Chi S}$.
For any $1\le p\le\infty$, the $p$-norm of $f$ is defined to be $\Norm f_p=(\sum_x|f(x)|^p)^{1/p}$.
And $\Norm f_\infty=\max_x|f(x)|$.

The degree of $f$, denoted by $\Deg(f)$, is the largest size of $S\subseteq[n]$ such that $\widehat f_S\ne0$.
For $0\le\varepsilon<1$, the $\varepsilon$-approximate degree of $f$, denoted by $\Deg_\varepsilon(f)$, is the smallest degree of any $p:\{0,1\}^n\to\Real$ such that $\Norm{f-p}_\infty\le\varepsilon$.
The block sensitivity of $f$ on input $x$, denoted by $\Bs(f,x)$, is the largest $k$ such that there are disjoint $B_1,\cdots,B_k\subseteq[n]$, $f(x)\ne f\LS x^{\oplus B_i}\RS$ for each $i\in[k]$.
Here for $x\in\{0,1\}^n$ and $S\subseteq[n]$, $x^{\oplus S}\in\{0,1\}^n$ satisfies $x^{\oplus S}_i=x_i\oplus1$ for $i\in S$ and $x^{\oplus S}_i=x_i$ for $i\notin S$.
The block sensitivity of $f$, denoted by $\Bs(f)$, is $\max_{x\in\{0,1\}^n}\Bs(f,x)$.

\bigskip

There is a dual characterization of the approximate degree.
A polynomial $p$ of degree $d$ which approximates $f$ provides a certificate that the approximate degree of $f$ is at most $d$.
Similarly, a dual polynomial for $f$ provides a certificate that the approximate degree of $f$ is \textit{at least} some value.
More precisely, the dual polynomial has the following properties.

\begin{lemma}[{\cite[Theorem 3.2]{She11}}]
\label{lem:dual_polynomial}
For a function $f:\{0,1\}^n\to\Real$ and $0\le\varepsilon<1$, if $\Deg_\varepsilon(f)\ge d$, there is a function $\psi:\{0,1\}^n\to\Real$ such that
\begin{itemize}
\item $\Norm\psi_1=1$ and $\IP f\psi\ge\varepsilon$.
\item $\widehat\psi_S=0$ for any $S\subseteq[n]$ of size smaller than $d$.
\end{itemize}
\end{lemma}

We are interested in the complexity measure that satisfies: for any Boolean function, any coordinate can be fixed s.t. the measure decreases by at most $1$.
The formal definition is as follows.

\begin{definition}
\label{def:entropic_measure}
Let $C(\cdot)$ be a complexity measure of the Boolean function $f:\{0,1\}^*\to\Real$.
We say that $C(\cdot)$ is entropic if for any set $J$ of coordinates, any function $f:\{0,1\}^J\to\Real$, and any coordinate $i\in J$, there is a $z_i\in\{0,1\}$ such that $C\LS f_{J\setminus\{i\},z_i}\RS\ge C(f)-1$.
\end{definition}

\begin{proposition}
\label{prop:entropic_transitivity}
Let $C(\cdot)$ be a complexity measure of the Boolean function $f:\{0,1\}^*\to\Real$.
If $C(\cdot)$ is entropic, for any set $J$ of coordinates, any function $f:\{0,1\}^J\to\Real$, and any $K\subset J$, there is a $z\in\{0,1\}^{J\setminus K}$ such that $C\LS f_{K,z}\RS\ge C(f)-(|J|-|K|)$.
\end{proposition}
\begin{proof}
Let $\ell=|J|-|K|$ and $J\setminus K=\LL i_1,\cdots,i_\ell\RL$.
Let $K^{(0)}=J$ and $K^{(j)}=K^{(j-1)}\setminus\{i_j\}$ for every $j\in[\ell]$.
So $K^{(\ell)}=K$.
For any $z\in\{0,1\}^{J\setminus K}$, let $g^{(0)}=f$ and
$$g^{(j)}=g^{(j-1)}_{K^{(j)},z_{i_j}}=f_{K^{(j)},z_{\{ i_1,\cdots i_j\}}}$$
for every $j\in[\ell]$.
So $g^{(\ell)}=f_{K,z}$
We prove that for every $j\in[\ell]$, there exists a $z_{J\setminus K^{(j)}}$ such that $C\LS g^{(j)}\RS\ge C(f)-j$ by induction on $j$.

The base case $j=0$ holds trivially.
Assume by the induction hypothesis that there exists a $z_{\{i_1,\cdots,i_{j-1}\}}$ such that $C\LS g^{(j-1)}\RS\ge C(f)-(j-1)$ for $j\in[\ell]$.
As $C(\cdot)$ is entropic, there exists a $z_{i_j}\in\{0,1\}$ such that $C\LS g^{(j)}\RS=C\LS g^{(j-1)}_{K^{(j)},z_{i_j}}\RS\ge C\LS g^{(j-1)}\RS-1\ge C(f)-j$.
\end{proof}

\begin{proposition}
\label{prop:deg_is_entropic}
The degree $\Deg(\cdot)$ is entropic.
\end{proposition}
\begin{proof}
Let $d=\Deg(f)$, there is an $I\subseteq J$ such that $|I|=d$ and $\hat f_I\ne0$.
For $i\in J$ and $z_i\in\{0,1\}$, let $g=f_{J\setminus\{i\},z_i}$.
The Fourier expansion yields
$$g=\sum_{S\subseteq J\setminus\{i\}}\LS\hat f_S+\hat f_{S\uplus\{i\}}(-1)^{z_i}\RS\Chi S.$$
So $\hat g_S=\hat f_S+\hat f_{S\uplus\{i\}}(-1)^{z_i}$ for any $S\subseteq J\setminus\{i\}$.
\begin{itemize}
\item If $i\in I$, we have $I\setminus\{i\}\subseteq J\setminus\{i\}$ and $\hat g_{I\setminus\{i\}}=\hat f_{I\setminus\{i\}}+\hat f_I(-1)^{z_i}$.
As $\hat f_I\ne0$, there exists a $z_i\in\{0,1\}$ such that $\hat g_{I\setminus\{i\}}\ne0$, and $\Deg(g)\ge|I\setminus\{i\}|=d-1$.
\item If $i\notin I$, we have $I\subseteq J\setminus\{i\}$ and $\hat g_I=\hat f_I+\hat f_{I\uplus\{i\}}(-1)^{z_i}$.
As $\hat f_I\ne0$, there exists a $z_i\in\{0,1\}$ such that $\hat g_I\ne0$, and $\Deg(g)\ge|I|=d$.
\end{itemize}
In conclusion, there is a $z_i\in\{0,1\}$ such that $\Deg\LS f_{J\setminus\{i\},z_i}\RS\ge d-1$.
\end{proof}

\begin{proposition}
\label{prop:bs_is_entropic}
The block sensitivity $\Bs(\cdot)$ is entropic.
\end{proposition}
\begin{proof}
Let $k=\Bs(f)$, there is an $x\in\{0,1\}^J$ and disjoint $B_1,\cdots,B_k$ such that $f(x)\ne f\LS x^{\oplus B_j}\RS$ for each $j\in[k]$.
For $i\in J$, we choose $z_i=x_i$, and we have $f_{J\setminus\{i\},z_i}\LS x_{J\setminus\{i\}}\RS=f(x)$.

For $j\in[k]$ such that $i\notin B_j$, we have $B_j\subseteq J\setminus\{i\}$, and
$$f_{J\setminus\{i\},z_i}\LS x_{J\setminus\{i\}}\RS\ne f_{J\setminus\{i\},z_i}\LS x^{\oplus B_j}_{J\setminus\{i\}}\RS$$
since $f(x)\ne f\LS x^{\oplus B_j}\RS$.
Among $j\in[k]$, the number of $B_j$ that do not contain $i$ is at least $k-1$ since $B_1,\cdots,B_k$ are disjoint.
Therefore, $\Bs\LS f_{J\setminus\{i\},z_i}\RS\ge k-1$.
\end{proof}

\subsection{Classical communication complexity}
\label{sec:classical_communication}

We will employ the model of classical communication complexity introduced by Yao \cite{Yao79}.
Let $F:\X\times\Y\to\{\pm1\}$ be a function with its input distributed between two parties: Alice knows $x\in\X$ and Bob knows $y\in\Y$.
Communication complexity studies the minimum number of bits they need to exchange in order to compute the function $F$.

% \paragraph{Deterministic communication.}
In the deterministic communication model, a protocol with communication cost $c$ will have Alice and Bob alternately sending $c$ classical bits by round.
Alice sends a bit $m_i\in\{0,1\}$ in the $i$-th round for odd $i$, and Bob sends a bit $m_i\in\{0,1\}$ in the $i$-th round for even $i$.
Therefore, $m_i$ is an arbitrary function of $(x,m_1\cdots m_{i-1})$ for odd $i$, and of $(y,m_1\cdots m_{i-1})$ for even $i$, where $(x,y)$ is the input.
The string $m_1\cdots m_i$ is called the \textit{transcript} of the first $i$ bits for each $i\in[c]$.
The protocol is said to compute $F$ if $(-1)^{m_c}=F(x,y)$ for every input $(x,y)\in\X\times\Y$.
We use $D(F)$ to denote the least communication cost of a deterministic protocol that computes $F$.

A typical characterization of a deterministic communication protocol is that it partitions the rectangle $\X\times\Y$ into disjoint sub-rectangles.
That is, $\X\times\Y=\biguplus_{m\in\{0,1\}^c}R_m=\biguplus_{m\in\{0,1\}^c}\X_m\times \Y_m$ such that for each $m\in\{0,1\}^c$, the transcript of the $c$ bits is $m$ on any input $(x,y)\in R_m$.

% \paragraph{Randomized communication.} In the randomized communication model, Alice and Bob are allowed to use private and public randomness.
% Let $A,B,P$ be finite random strings, where $A$ and $B$ are private randomness for Alice and Bob, respectively, and $P$ are the public randomness shared between them.
% Beside $x\in\X$, Alice takes $A$ and $P$ as input, and Bob takes $B$ and $P$ beside $y\in\Y$. 
% For $0\le\varepsilon<1$, a randomized communication protocol is said to compute $F$ with error $\varepsilon$ if the probability of $(-1)^{m_c}=F(x,y)$ over the random variables $A,B,P$ is at least $1-\varepsilon$ on any $(x,y)\in\X\times\Y$.
% We use $R_\varepsilon(F)$ to denote the least communication cost of a randomized protocol that computes $F$ with error $\varepsilon$.

\subsection{Quantum communication complexity}
\label{sec:quantum_communication}

Since Yao introduced quantum communication complexity in 1993 \cite{Yao93}, there have been several equivalent ways to describe a $2$-party quantum communication protocol.
Our description follows Lee and Shraibman \cite{LS09}.
The state of a quantum communication protocol can be represented as a vector in a Hilbert space $H_A\otimes C\otimes H_B$.
Here, $H_A,H_B$ are Hilbert spaces of arbitrary finite dimension that represent \textit{workspaces} of Alice and Bob, respectively.
The Hilbert space $C$ is $2$-dimensional, and it stands for a $1$-qubit channel.
We assume that $H_A$ contains a register to hold the input of Alice, and similarly for $H_B$.

In the model without prior entanglement, the initial state of a quantum communication protocol on input $(x,y)$ is the vector $\ket{x,0}_{H_A}\ket0_C\ket{y,0}_{H_B}$.
With prior entanglement, the initial state is a vector of the form $\sum_w\alpha_w\ket{x,w}_{H_A}\ket0_C\ket{y,w}_{H_B}$, where the coefficients $\alpha_w$ are arbitrary complex numbers satisfying $\sum_w|\alpha_w|^2=1$.

We assume that Alice and Bob \textit{speak} alternately.
On Alice's turn, she applies an arbitrary unitary transformation of the form $U_{H_AC}\otimes I_B$, which acts as the identity on $H_B$.
Similarly, on Bob's turn, he applies a transformation of the form $I_A\otimes U_{H_BC}$.
At the end of a $t$-round protocol, we project the final state onto the subspace $H_A\otimes\ket1\otimes H_B$.
Denoting the length of this projection by $p$, the protocol outputs $-1$ with probability $p^2$, and output $1$ otherwise.
We say that the communication cost of the protocol is $t$. 

We assume that Alice and Bob exchange quantum messages through a quantum channel, with no intermediate measurements are allowed.
It is worth noting that when prior entanglement is allowed, the players can exchange classical bits to teleport quantum states, doubling the total communication cost.
This setting is referred to as the Cleve-Buhrman model \cite{CB97}.

\bigskip

For a $2$-argument function $F:\X\times\Y\to\{\pm1\}$ and $0\le\varepsilon<1$, a quantum communication protocol is said to compute $F$ with error $\varepsilon$ if it outputs $F(x,y)$ with probability at least $1-\varepsilon$ on any input $(x,y)\in\X\times\Y$.
Let $Q_\varepsilon(F)$ denote the least communication cost of a quantum protocol without prior entanglement that computes $F$ with error $\varepsilon$.
Define $Q^*_\varepsilon(F)$ analogously for protocols with prior entanglement.

\paragraph{Generalized discrepancy bound.}
The generalized discrepancy method is a useful technique for proving lower bounds on quantum communication complexity, regardless of prior entanglement.
This technique is originally discovered by Klauck \cite{Kla07} and Razborov \cite{Raz03}.
The generalized discrepancy bound can be established through multiple approaches, including methods based on factorization norms \cite{LS07} and XOR games \cite[Section 5.3]{LS09}.
The following is an adaptation by Sherstov \cite{She11}.

\begin{theorem}[{\cite[Theorem 2.8]{She11}}]
\label{thm:discrepancy_bound}
Let $U,V$ be finite sets and $F:U\times V\to\{\pm1\}$ be a given function.
Let $\Psi=(\Psi_{uv})_{u\in U,v\in V}$ be any real matrix.
For each $\varepsilon>0$,
$$Q^*_\varepsilon(F)=\Omega\LS\log\frac{\IP F\Psi-2\varepsilon\Norm\Psi_1}{3\Norm\Psi\sqrt{|U|\cdot|V|}}\RS.$$
\end{theorem}

\subsection{Hybrid classical-quantum communication}
\label{sec:hybrid_communication}

In this work, we are concerned with hybrid classical-quantum communication protocols.
A hybrid classical-quantum communication protocol consists of two phases: a classical phase followed by a quantum phase.
The protocol begins with the classical phase, during which Alice and Bob alternately and \textit{deterministically} send classical bits and implement classical local computation. In the subsequent quantum phase, Alice and Bob exchange qubits and implement quantum computation.
We assume that the players are allowed to share prior entanglement.
It is worth noticing that, in the first phase, both players are implementing classical computation.
They do not touch shared entanglement and thus cannot exchange quantum messages via quantum teleportation.

Let $\Protocol(c,q)$ be a hybrid protocol with a $c$-bit classical communication in the first phase followed by a $q$-qubit quantum communication in the second phase.
Without loss of generality, we may assume that $q=o(c)$ since otherwise we may simulate the classical communication by the quantum communication, which only doubles the communication cost.

Suppose that the input to $\Protocol$ are drawn from $\X\times\Y$.
For a function $F:\X\times\Y\to\{\pm1\}$ and $0\le\varepsilon<1$, $\Protocol$ is said to compute $F$ with error $\varepsilon$ if it outputs $F(x,y)$ with probability at least $1-\varepsilon$ on any input $(x,y)\in\X\times\Y$.

After the classical communication phase of $c$ bits, as discussed in \Cref{sec:classical_communication}, $\X\times\Y$ is partitioned into disjoint rectangles in the form of $\X\times\Y=\biguplus_{m\in\{0,1\}^c}R_m$, where $R_m$ contains all inputs on which the transcript is $m$ for each $m\in\{0,1\}^c$.
Let $F_m$ denote the function $F$ with its input restricted to $R_m$.
The following proposition follows by the definition.

\begin{proposition}
\label{prop:qcc_on_rectangle}
If $\Protocol(c,q)$ computes $F$ with error $\varepsilon$, $Q^*_\varepsilon(F_m)\le q$ for each $m\in\{0,1\}^c$.
\end{proposition}

\subsection{Composed functions}

An important family of communication functions are composed functions of the form
$$f\circ G^n(x,y)=f\LS G^n(x,y)\RS=f\LS G(x_1,y_1),\cdots,G(x_n,y_n)\RS,$$
where $f:\{0,1\}^n\to\{\pm1\}$ is the outer function, and $G:\X\times\Y\to\{0,1\}$ is the inner function, also known as the gadget.
The size of the gadget is defined as $\log\min\{|\X|,|\Y|\}$.
Lifting theorems typically establish lower bounds of $D(f\circ G^n)$, $R_\varepsilon(f\circ G^n)$, or $Q_\varepsilon(f\circ G^n)$ in terms of some complexity measure of $f$, such as the deterministic/randomized query complexity and the approximate degree.

In this work, we choose the gadget $G:\Lambda\times\Lambda\to\{0,1\}$ as the inner product function over the Boolean domain, where $\Lambda=\{0,1\}^b$ and $b=20\log n$.
We may focus on a set $J\subseteq[n]$ of coordinates, then $x_J=(x_i)_{i\in J}$ and $G^J(x,y)=\LS G(x_i,y_i)\RS_{i\in J}$.

\begin{definition}[Block-wise density]
\label{def:blockwise_density}
Let $J$ be a set of coordinates and $0<\delta\le1$.
A random variable $X\in\Lambda^J$ is $\delta$-dense if for every $I\subseteq J$, it holds that
$$H_\infty(X_I)=\log\frac1{\max_{x\in\Lambda^I}\Dist{X_I}(x)}\ge\delta\cdot b\cdot|I|.$$
That is, $\Dist{X_I}(x)\le2^{-\delta b|I|}$ for all $I\subseteq J$ and $x\in\Lambda^I$.
\end{definition}

\begin{lemma}[{\cite[Lemma 13]{GLMWZ16}}]
\label{lem:pointwise_uniform}
Let $J$ be a set of coordinates and $X,Y\in\Lambda^J$ be independent random variables which are $\delta_X$-dense and $\delta_Y$-dense, respectively.
If $\delta_X+\delta_Y\ge1.4$, it holds that
$$\Dist{G^J(X,Y)}(z)\in\LS1\pm n^{-2}\RS2^{-|J|},\qquad\forall\ z\in\{0,1\}^J.$$
\end{lemma}

We give a proof of the above lemma in \Cref{sec:appendix} for completeness.
It is a key lemma in a line of work on classical query-to-communication lifting \cite{GLMWZ16,GPW20,CFKMP19,CFKMP21}.
And we will combine it with approaches from another line of work on lifting the approximate degree to quantum communication complexity \cite{She11,SZ09,LZ10}.
\section{Quantum Communication Lower Bounds on Dense Rectangles}

Here we adopt the generalized discrepancy method in~\Cref{thm:discrepancy_bound} to prove the hybrid quantum communication complexity. 
For a composed function $F=f\circ G^n$, a series of works \cite{She11,SZ09,LZ10} develop conditions on the gadget $G$ under which the approximate degree of any outer function $f$ can be lifted to the quantum communication complexity of $F$.
As stated in \Cref{thm:discrepancy_bound}, any witness matrix $\Psi$ of the same dimension as $F$ reveals a lower bound on the quantum communication complexity of $F$.
A natural choice is to set $\Psi=\psi\circ G^n$, where $\psi$ is a dual polynomial of $f$ as stated in \Cref{lem:dual_polynomial}.
As $\psi=\sum_{S\subseteq[n]}\widehat\psi_S\Chi S$, we have $\Psi=\sum_{S\subseteq[n]}\widehat\psi_SM_S$, where $M_S=\Chi S\circ G^n$. It has been proved that the matrices $M_S$ are \textit{strongly orthogonal} (that is, $M_S M_T^\Trans=M_S^\Trans M_T=0$ for $S\ne T$) to each other for nice gadgets, such that it enables us to prove tight lower bounds.

Let $n$ be sufficiently large.
Recall that our gadget $G:\Lambda\times\Lambda\to\{0,1\}$ is the inner product function over the Boolean domain, where $\Lambda=\{0,1\}^b$ and $b=20\log n$.

When the input of $f\circ G^n$ is restricted to a rectangle, the sub-matrices of matrices $M_S$ obtained by restricting $M_S$ to a rectangle are no longer strongly orthogonal to each other.
The following result is our main technical result, which says that the matrix $\Psi=\psi\circ G^n$ restricted to the same rectangle can still be used to show strong quantum communication lower bounds, as long as the size of the gadget $b=\Theta(\log n)$ is sufficiently large (while in \cite{She11,SZ09,LZ10}, the size of the gadget can be constant), and the rectangle satisfies certain density properties.

\begin{theorem}
\label{thm:lifting_on_rectangle}
Let $J\subseteq[n]$ be a set of coordinates and $\ell$ be an arbitrary finite number.
For sets $U,V\subseteq\Lambda^J\times\{0,1\}^\ell$, let the joint random variables $(X,A)$ be uniformly distributed over $U$ and $(Y,B)$ be uniformly distributed over $V$, where $X, Y\in\Lambda^J$.
For any function $f:\{0,1\}^J\to\{\pm1\}$, define $F:U\times V\to\{\pm1\}$ such that
$$F(u,v)=f(G^J(x,y)),\qquad\forall\ u=(x,a)\in U,v=(y,b)\in V.$$
If $X$ and $Y$ are both $0.99$-dense,
$$Q^*_{0.1}(F)=\Omega\LS\Deg_{1/3}(f)\cdot b\RS.$$
\end{theorem}

\begin{proof}

Let $R$ be the rectangle $U\times V$, $m=|J|$ and $d=\Deg_{1/3}(f)$.
Let $\psi$ be the dual polynomial which certifies that $\Deg_{1/3}(f)\ge d$ as in \Cref{lem:dual_polynomial}.
We have
\begin{itemize}
\item $\Norm\psi_1=1$ and $\IP f\psi\ge1/3$.
\item $\widehat\psi_S=0$ for any $S\subseteq[n]$ of size smaller than $d$.
\end{itemize}
The hardness of computing $F$ with error $0.1$ can be proved by choosing the matrix $\Psi:U\times V\to\Real$ that satisfies $\Psi(u,v)=\frac{2^m}{|R|}\psi(G^J(x,y))$ for every $u=(x,a)\in U,v=(y,b)\in V$, and applying \Cref{thm:discrepancy_bound}.
The task is to bound $\Norm\Psi_1$, $\IP F\Psi$, and $\Norm\Psi$.

\begin{lemma}
\label{lem:bound_one_norm_and_ip}
$\Norm\Psi_1\le1.01$ and $\IP F\Psi\ge0.31$.
\end{lemma}

\begin{lemma}
\label{lem:bound_spectral_norm}
$\Norm\Psi\sqrt{|R|}\le n^{-1.1d}$.
\end{lemma}

Combining \Cref{lem:bound_one_norm_and_ip}, \Cref{lem:bound_spectral_norm} with \Cref{thm:discrepancy_bound}, we have
\begin{equation*}
Q^*_{0.1}(F)=\Omega\LS\log\frac{\IP F\Psi-0.2\Norm\Psi_1}{3\Norm\Psi\sqrt{|R|}}\RS=\Omega(d\cdot b). 
\qedhere
\end{equation*}
\end{proof}

The rest of this section proves \Cref{lem:bound_one_norm_and_ip} and \Cref{lem:bound_spectral_norm}.

\begin{proof}[Proof of \Cref{lem:bound_one_norm_and_ip}]
By the choice of $\Psi$,
\begin{align*}
\Norm\Psi_1 & =\sum_{\substack{(x,a)\in U \\ (y,b)\in V}}\frac{2^m}{|R|}\Abs{\psi(G^J(x,y))} \\
& =\sum_{z\in\{0,1\}^J}\frac{2^m}{|R|}|\psi(z)|\cdot\Abs{\LL((x,a),(y,b))\in R:G^J(x,y)=z\RL}.
\end{align*}
Note that $\Dist{G^J(X,Y)}(z)=\frac{\Abs{\LL((x,a),(y,b))\in R:G^J(x,y)=z\RL}}{|R|}$ as the joint random variable $((X,A),(Y,B))$ is uniformly distributed over $R$.
Then
\begin{equation}
% label{equ:bound_one_norm}
\nonumber\Norm\Psi_1=\sum_{z\in\{0,1\}^J}2^m|\psi(z)|\cdot\Dist{G^J(X,Y)}(z)\le1.01\Norm\psi_1=1.01.
\end{equation}
The inequality is because $X,Y$ are $0.99$-dense.
By \Cref{lem:pointwise_uniform}, 
$$\Dist{G^J(X,Y)}(z)\le\LS1+n^{-2}\RS2^{-m}\le1.01\cdot2^{-m},$$
for any $z\in\{0,1\}^J$.
The inner product term can be bounded in a similar way:
\begin{align}
\nonumber\IP F\Psi & =\sum_{\substack{(x,a)\in U \\ (y,b)\in V}}f(G^J(x,y))\cdot\frac{2^m}{|R|}\psi(G^J(x,y)) \\
\nonumber & =\sum_{z\in\{0,1\}^J}2^mf(z)\psi(z)\cdot\Dist{G^J(X,Y)}(z) \\
\nonumber & \ge0.99\IP f\psi-0.02\Norm\psi_1\ge0.31.
% \label{equ:bound_inner_product}
\qedhere
\end{align}
\end{proof}

\begin{proof}[Proof of \Cref{lem:bound_spectral_norm}]
To bound the spectral norm, we have $\Norm\Psi\le\LS\Tr\LS\Psi\Psi^\Trans\RS^p\RS^\frac1{2p}$ for any $p\ge1$.
We choose $p=2$.
As $\widehat\psi_S=0$ for every $S\subseteq J$ such that $|S|<d$, $\psi=\sum_{S\subseteq J:|S|\ge d}\widehat\psi_S\Chi S$.
Then $\Psi=\frac{2^m}{|R|}\sum_{S\subseteq J:|S|\ge d}\widehat\psi_S M_S$, where $M_S(u,v)=\Chi S(G^J(x,y))$ for every $u=(x,a)\in U,v=(y,b)\in V$.
\begin{align}
\nonumber\Tr\LS\Psi\Psi^\Trans\RS^2 & =\frac{2^{4m}}{|R|^4}\sum_{\substack{S_1,T_1,S_2,T_2\subseteq J \\ |S_1|,|T_1|,|S_2|,|T_2|\ge d}}\widehat\psi_{S_1}\widehat\psi_{T_1}\widehat\psi_{S_2}\widehat\psi_{T_2}\Tr\LS M_{S_1}M_{T_1}^\Trans M_{S_2}M_{T_2}^\Trans\RS \\
& \le\frac1{|R|^4}\sum_{S_1,T_1,S_2,T_2}\Abs{\Tr\LS M_{S_1}M_{T_1}^\Trans M_{S_2}M_{T_2}^\Trans\RS}.
\label{equ:bound_trace}
\end{align}
The inequality is because that $|\widehat\psi_S|\le2^{-m}\Norm\psi_1\le2^{-m}$ for any $S\subseteq J$.
The following states that each term of the summation in \Cref{equ:bound_trace} can be bounded in terms of $|S_1|+|T_1|+|S_2|+|T_2|$.

\begin{proposition}
\label{prop:bound_each_trace}
For any $S_1,T_1,S_2,T_2\subseteq J$,
$$\frac1{|R|^2}\Abs{\Tr\LS M_{S_1}M_{T_1}^\Trans M_{S_2}M_{T_2}^\Trans\RS}\le2^{-0.11b(|S_1|+|T_1|+|S_2|+|T_2|)}.$$
\end{proposition}

We prove \Cref{prop:bound_each_trace} after showing that it implies \Cref{lem:bound_spectral_norm}.
According to \Cref{equ:bound_trace},
\begin{align}
\nonumber\Norm\Psi\sqrt{|R|}\le\LS\Tr\LS\Psi\Psi^\Trans\RS^2\RS^{1/4}\sqrt{|R|} & \le\LS\sum_{\substack{S_1,T_1,S_2,T_2\subseteq J \\ |S_1|,|T_1|,|S_2|,|T_2|\ge d}}\frac1{|R|^2}\Abs{\Tr\LS M_{S_1}M_{T_1}^\Trans M_{S_2}M_{T_2}^\Trans\RS}\RS^{1/4} \\
\nonumber & \le\LS\sum_{\substack{S_1,T_1,S_2,T_2\subseteq J \\ |S_1|,|T_1|,|S_2|,|T_2|\ge d}}2^{-0.11b(|S_1|+|T_1|+|S_2|+|T_2|)}\RS^{1/4} \\[1ex]
\nonumber & =\sum_{S\subseteq J:|S|\ge d}2^{-0.11b|S|}\le n^{-1.1d},
% label{equ:bound_spectral_norm}
\qedhere
\end{align}
\end{proof}

\begin{proof}[Proof of \Cref{prop:bound_each_trace}]
The trace term $\Tr\LS M_{S_1}M_{T_1}^\Trans M_{S_2}M_{T_2}^\Trans\RS$ can be expanded as
$$\sum_{\substack{(x_1,a_1),(x_2,a_2)\in U \\ (y_1,b_1),(y_2,b_2)\in V}}\Chi{S_1}(G^J(x_1,y_1))\Chi{T_1}(G^J(x_2,y_1))\Chi{S_2}(G^J(x_2,y_2))\Chi{T_2}(G^J(x_1,y_2)).$$
For any $I,K\subseteq J$, $x\in\Lambda^I,y\in\Lambda^K$, and $S\subseteq I\cap K$, define $\Chi S(x,y)=(-1)^{\IP{x_S}{y_S}}$.
By replacing the summation over $U,V$ with a summation over the entire $\Lambda^J$ and taking the corresponding probabilities, $\frac1{|R|^2}\Tr\LS M_{S_1}M_{T_1}^\Trans M_{S_2}M_{T_2}^\Trans\RS$ is equal to
\begin{align*}
& \sum_{x_1,x_2,y_1,y_2\in\Lambda^J}\Dist X(x_1)\Dist X(x_2)\Dist Y(y_1)\Dist Y(y_2)\cdot\Chi{S_1}(x_1,y_1)\Chi{T_1}(x_2,y_1)\Chi{S_2}(x_2,y_2)\Chi{T_2}(x_1,y_2) \\
= & \sum_{\substack{x_1\in\Lambda^{S_1\cup T_2},x_2\in\Lambda^{S_2\cup T_1} \\ y_1\in\Lambda^{S_1\cup T_1},y_2\in\Lambda^{S_2\cup T_2}}}\Dist{X_{S_1\cup T_2}}(x_1)\Dist{X_{S_2\cup T_1}}(x_2)\Dist{Y_{S_1\cup T_1}}(y_1)\Dist{Y_{S_2\cup T_2}}(y_2)\cdot\Matrix_{x_1x_2,y_1y_2} \\
= & \sum_{\substack{x_1x_2\in\Lambda^{S_1\cup T_2}\times\Lambda^{S_2\cup T_1} \\ y_1y_2\in\Lambda^{S_1\cup T_1}\times\Lambda^{S_2\cup T_2}}}\LS\Dist{X_{S_1\cup T_2}}\otimes\Dist{X_{S_2\cup T_1}}\RS(x_1x_2)\LS\Dist{Y_{S_1\cup T_1}}\otimes\Dist{Y_{S_2\cup T_2}}\RS(y_1y_2)\Matrix_{x_1x_2,y_1y_2},
\end{align*}
where we define the matrix $\Matrix\in\{\pm1\}^{\LS\Lambda^{S_1\cup T_2}\times\Lambda^{S_2\cup T_1}\RS\times\LS\Lambda^{S_1\cup T_1}\times\Lambda^{S_2\cup T_2}\RS}$ by
$$\Matrix_{x_1x_2,y_1y_2}=\Chi{S_1}(x_1,y_1)\Chi{T_1}(x_2,y_1)\Chi{S_2}(x_2,y_2)\Chi{T_2}(x_1,y_2).$$
Therefore,
\begin{align*}
\frac1{|R|^2}\Abs{\Tr\LS M_{S_1}M_{T_1}^\Trans M_{S_2}M_{T_2}^\Trans\RS} & =\Abs{\LS\Dist{X_{S_1\cup T_2}}\otimes\Dist{X_{S_2\cup T_1}}\RS^\Trans\Matrix\LS\Dist{Y_{S_1\cup T_1}}\otimes\Dist{Y_{S_2\cup T_2}}\RS} \\
& \le\Norm{\LS\Dist{X_{S_1\cup T_2}}\otimes\Dist{X_{S_2\cup T_1}}\RS^\Trans\Matrix}\cdot\Norm{\Dist{Y_{S_1\cup T_1}}\otimes\Dist{Y_{S_2\cup T_2}}},
\end{align*}
For the second term, $\Norm{\Dist{Y_{S_1\cup T_1}}\otimes\Dist{Y_{S_2\cup T_2}}}\le2^{-0.99b(|S_1\cup T_1|+|S_2\cup T_2|)/2}$ as $Y$ is $0.99$-dense.
For the first term, $\Norm{\LS\Dist{X_{S_1\cup T_2}}\otimes\Dist{X_{S_2\cup T_1}}\RS^\Trans\Matrix}^2$ equals to
\begin{align}
\nonumber & \sum_{y_1y_2\in\Lambda^{S_1\cup T_1}\times\Lambda^{S_2\cup T_2}}\LS\sum_{x_1x_2\in\Lambda^{S_1\cup T_2}\times\Lambda^{S_2\cup T_1}}\LS\Dist{X_{S_1\cup T_2}}\otimes\Dist{X_{S_2\cup T_1}}\RS(x_1x_2)\Matrix_{x_1x_2,y_1y_2}\RS^2 \\
\nonumber= & \sum_{y_1y_2}\sum_{x_1x_2,x'_1x'_2}\LS\Dist{X_{S_1\cup T_2}}\otimes\Dist{X_{S_2\cup T_1}}\RS(x_1x_2)\LS\Dist{X_{S_1\cup T_2}}\otimes\Dist{X_{S_2\cup T_1}}\RS(x'_1x'_2)\Matrix_{x_1x_2,y_1y_2}\Matrix_{x'_1x'_2,y_1y_2} \\
\nonumber= & \sum_{x_1x_2,x'_1x'_2}\LS\Dist{X_{S_1\cup T_2}}\otimes\Dist{X_{S_2\cup T_1}}\RS(x_1x_2)\LS\Dist{X_{S_1\cup T_2}}\otimes\Dist{X_{S_2\cup T_1}}\RS(x'_1x'_2)\sum_{y_1y_2}\Matrix_{x_1x_2,y_1y_2}\Matrix_{x'_1x'_2,y_1y_2}. \\
\le & \sum_{x_1x_2,x'_1x'_2}2^{-0.99b(|S_1\cup T_2|+|S_2\cup T_1|)\cdot2}\Abs{\sum_{y_1y_2}\Matrix_{x_1x_2,y_1y_2}\Matrix_{x'_1x'_2,y_1y_2}}.
\label{equ:bound_vector_norm}
\end{align}
The inequality is because $X$ is $0.99$-dense.
For any $x_1x_2,x'_1x'_2\in\Lambda^{S_1\cup T_2}\times\Lambda^{S_2\cup T_1}$, by setting $w_1=x_1\oplus x'_1\in\Lambda^{S_1\cup T_2}$ and $w_2=x_2\oplus x'_2\in\Lambda^{S_2\cup T_1}$,
\begin{align*}
& \sum_{y_1y_2\in\Lambda^{S_1\cup T_1}\times\Lambda^{S_2\cup T_2}}\Matrix_{x_1x_2,y_1y_2}\Matrix_{x'_1x'_2,y_1y_2} \\
= & \sum_{y_1y_2}\Chi{S_1}(w_1,y_1)\Chi{T_1}(w_2,y_1)\Chi{S_2}(w_2,y_2)\Chi{T_2}(w_1,y_2) \\
= & \sum_{y_1y_2}\Chi{A_1}(w_1\oplus w_2,y_1)\Chi{B_1}(w_1,y_1)\Chi{C_1}(w_2,y_1)\Chi{A_2}(w_1\oplus w_2,y_2)\Chi{B_2}(w_2,y_2)\Chi{C_2}(w_1,y_2),
\end{align*}
where $A_1=S_1\cap T_1,B_1=S_1\setminus T_1,C_1=T_1\setminus S_1$ are disjoint, and $A_2=S_2\cap T_2,B_2=S_2\setminus T_2,C_2=T_2\setminus S_2$ are disjoint.
Let $\Gamma_A(w)$ denote $\sum_{y\in\Lambda^A}\Chi A(w,y)$.
We have
$$\Gamma_A(w)=\left\{
\begin{aligned}
& 2^{b|A|}, && \text{if }w_A=0. \\
& 0, && \text{otherwise}.
\end{aligned}
\right.$$
And
$$\sum_{y_1y_2}\Matrix_{x_1x_2,y_1y_2}\Matrix_{x'_1x'_2,y_1y_2}=\Gamma_{A_1}(w_1\oplus w_2)\Gamma_{B_1}(w_1)\Gamma_{C_1}(w_2)\Gamma_{A_2}(w_1\oplus w_2)\Gamma_{B_2}(w_2)\Gamma_{C_2}(w_1).$$
We say that $w_1w_2$ is valid if
$$(w_1)_{A_1\cup A_2}=(w_2)_{A_1\cup A_2},(w_1)_{B_1\cup C_2}=0,(w_2)_{C_1\cup B_2}=0.$$
Then,
\begin{equation}
\label{equ:bound_valid_term}
\sum_{y_1y_2}\Matrix_{x_1x_2,y_1y_2}\Matrix_{x'_1x'_2,y_1y_2}=\left\{
\begin{aligned}
& 2^{b(|S_1\cup T_1|+|S_2\cup T_2|)}, && \text{if }w_1w_2\text{ is valid}. \\
& 0, && \text{otherwise}.
\end{aligned}
\right.
\end{equation}

\begin{proposition}
\label{prop:count_valid}
The number of valid $w_1w_2\in\Lambda^{S_1\cup T_2}\times\Lambda^{S_2\cup T_1}$ is at most
$$2^{b(|S_1\cup T_2|+|S_2\cup T_1|)}\cdot2^{-b(|S_1|+|T_1|+|S_2|+|T_2|)/4}.$$
\end{proposition}
The proof of \Cref{prop:count_valid} is postponed.
The number of $x_1x_2,x'_1x'_2\in\Lambda^{S_1\cup T_2}\times\Lambda^{S_2\cup T_1}$ s.t.
$$w_1w_2\text{ is valid}\qquad(\text{where }w_1=x_1\oplus x'_1\text{ and }w_2=x_2\oplus x_2)$$
will be $2^{b(|S_1\cup T_2|+|S_2\cup T_1|)}$ times the number of valid values in $\Lambda^{S_1\cup T_2}\times\Lambda^{S_2\cup T_1}$.
Hence by combining \Cref{equ:bound_vector_norm}, \Cref{equ:bound_valid_term} with \Cref{prop:count_valid}, $\Norm{\LS\Dist{X_{S_1\cup T_2}}\otimes\Dist{X_{S_2\cup T_1}}\RS^\Trans\Matrix}^2$ is at most
$$2^{-0.99b(|S_1\cup T_2|+|S_2\cup T_1|)\cdot2}\cdot2^{2b(|S_1\cup T_2|+|S_2\cup T_1|)}\cdot2^{-b(|S_1|+|T_1|+|S_2|+|T_2|)/4}\cdot2^{b(|S_1\cup T_1|+|S_2\cup T_2|)}.$$
In conclusion,
\begin{align*}
\frac1{|R|^2}\Abs{\Tr\LS M_{S_1}M_{T_1}^\Trans M_{S_2}M_{T_2}^\Trans\RS} & \le\Norm{\LS\Dist{X_{S_1\cup T_2}}\otimes\Dist{X_{S_2\cup T_1}}\RS^\Trans\Matrix}\cdot\Norm{\Dist{Y_{S_1\cup T_1}}\otimes\Dist{Y_{S_2\cup T_2}}} \\[1ex]
& \le2^{0.01b(|S_1\cup T_2|+|S_2\cup T_1|)}\cdot2^{-b(|S_1|+|T_1|+|S_2|+|T_2|)/8}\cdot2^{0.005b(|S_1\cup T_1|+|S_2\cup T_2|)} \\[1ex]
& \le2^{-0.11b(|S_1|+|T_1|+|S_2|+|T_2|)},
\end{align*}
which is as required by \Cref{prop:bound_each_trace}.
\end{proof}

\begin{proof}[Proof of \Cref{prop:count_valid}]
The number of coordinates of $w_1w_2$ is $|S_1\cup T_2|+|S_2\cup T_1|$, and there are $b$ bits on each coordinates.
Some coordinates are fixed by the validity condition of $w_1w_2$, while the others are totally free.
The number of fixed coordinates is at least $|S_1\cup T_1\cup S_2\cup T_2|$.
See \Cref{fig:valid} for an illustration.
A formal proof is given below.

\begin{figure*}[ht]
\centering
\includegraphics[width=0.75\linewidth]{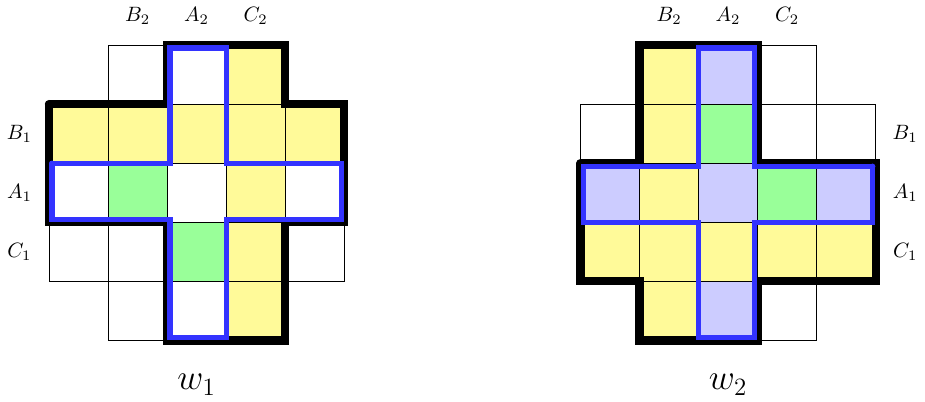}
\caption{
An illustration of valid $w_1w_2\in\Lambda^{(A_1\uplus B_1)\cup(A_2\uplus C_2)}\times\Lambda^{(A_1\uplus C_1)\cup(A_2\uplus B_2)}$.
The coordinates outlined by the thick black line are those of $w_1w_2$.
A valid $w_1w_2$ satisfies: $(w_1)_{A_1\cup A_2}=(w_2)_{A_1\cup A_2}$, $(w_1)_{B_1\cup C_2}=0$, and $(w_2)_{C_1\cup B_2}=0$.
The coordinates outlined by the blue line require that $w_1$ and $w_2$ be identical.
So on the coordinates within the yellow and green area, $w_1w_2$ must be all $0$.
On the coordinates in the blue area, $w_2$ must agree with $w_1$.
For the rest of the coordinates, the choice can be arbitrary.
}
\label{fig:valid}
\end{figure*}

The number of valid $w_1w_2\in\Lambda^{S_1\cup T_2}\times\Lambda^{S_2\cup T_1}$ is
\begin{align*}
& \sum_{w_1w_2\in\Lambda^{S_1\cup T_2}\times\Lambda^{S_2\cup T_1}}\Indicator\LM(w_1)_{A_1\cup A_2}=(w_2)_{A_1\cup A_2}\wedge(w_1)_{B_1\cup C_2}=0\wedge(w_2)_{C_1\cup B_2}=0\RM \\
= & \sum_{w_1\in\Lambda^{S_1\cup T_2}}\Indicator\LM(w_1)_{B_1\cup C_2}=0\RM\sum_{w_2\in\Lambda^{S_2\cup T_1}}\Indicator\LM(w_2)_{A_1\cup A_2}=(w_1)_{A_1\cup A_2}\wedge(w_2)_{C_1\cup B_2}=0\RM.
\end{align*}
Fix any $w_1\in\Lambda^{S_1\cup T_2}$.
As $A_1\cup A_2\cup C_1\cup B_2=(A_1\uplus C_1)\cup(A_2\uplus B_2)=S_2\cup T_1$, the number of $w_2\in\Lambda^{S_2\cup T_1}$ satisfying $(w_2)_{A_1\cup A_2}=(w_1)_{A_1\cup A_2}$ and $(w_2)_{C_1\cup B_2}=0$ is at most $1$.
The number of valid $w_1w_2$ is at most
$$\sum_{w_1\in\Lambda^{S_1\cup T_2}}\Indicator\LM(w_1)_{B_1\cup C_2}=0\RM=2^{b\LS|S_1\cup T_2|-|B_1\cup C_2|\RS} =2^{b\LS|S_1\cup T_2|+|S_2\cup T_1|\RS}\cdot2^{-b\LS|B_1\cup C_2|+|S_2\cup T_1|\RS}.$$
And $|B_1\cup C_2|+|S_2\cup T_1|\ge|B_1\cup C_2\cup S_2\cup T_1|=|S_1\cup T_1\cup S_2\cup T_2|\ge\frac{|S_1|+|T_1|+|S_2|+|T_2|}4$.
\end{proof}
\section{A Decision Tree Yielding Dense Rectangles}

We apply the techniques developed in a series of works on classical query-to-communication lifting \cite{GPW20, CFKMP19, CFKMP21}.
Those works consider the composed function $f\circ G^n$ with arbitrary outer function $f:\{0,1\}^n\to\{\pm1\}$ and some size $\Theta(\log n)$-sized gadget $G$.
The gadget is chosen to be the index function \cite{GPW20}, the inner product function over the Boolean domain \cite{CFKMP19}, or any function with low discrepancy \cite{CFKMP21}.
A decision tree for $f$ is constructed based on the communication protocol for $f\circ G^n$ by keeping the rectangles being dense.

\bigskip

Let $n$ be sufficiently large.
Recall that our gadget $G:\Lambda\times\Lambda\to\{0,1\}$ is the inner product function over the Boolean domain, where $\Lambda=\{0,1\}^b$ and $b=20\log n$.

\begin{definition}[Density of rectangles]
\label{def:rectangle_density}
Let $J\subseteq[n]$ be a set of coordinates and $z\in\{0,1\}^{\bar J}$.
For sets $U,V\subseteq\Lambda^n$, let the random variables $X$ be uniformly distributed over $U$ and $Y$ be uniformly distributed over $V$.
The rectangle $R=U\times V$ is dense on $(J,z)$ if
\begin{itemize}
\item $G^{\bar J}(x,y)=z$ for any $(x,y)\in R$.
\item $X_J$ and $Y_J$ are both $0.99$-dense.
\end{itemize}
\end{definition}

Throughout this section, we consider a hybrid classical-quantum communication protocol $\Protocol$ that takes the input from $\Lambda^n\times\Lambda^n$, and exchages $c$ bits in the classical communication phase.
As discussed in \Cref{sec:hybrid_communication}, $\Lambda^n\times\Lambda^n$ is partitioned into disjoint rectangles in the form of $\Lambda^n\times\Lambda^n=\biguplus_{m\in\{0,1\}^c}R_m$ such that for each $m\in\{0,1\}^c$, the transcript is $m$ on any input $(x,y)\in R_m$.

The following lemma shows that there is a decision tree such that a rectangle within some $R_m$ can be found, which is dense on the unqueried coordinates and consistent with the queried ones.
The analysis follows the deterministic query-to-communication lifting in \cite{CFKMP19}.
Initially $J=[n]$ is the set of unqueried coordinates and $R=\Lambda^n\times\Lambda^n$.
Whenever a bit is transmitted in $\Protocol$, $R$ is partitioned into two rectangles, and we retain the larger one.
Intuitively whenever $R$ is not dense on $J$, we can make queries to some $I\subset J$ and find a rectangle $R'\subseteq R$ such that $R'$ becomes dense on $J\setminus I$ and $G^I$ is consistent with the query answer on $R'$.

\begin{lemma}
\label{lem:construct_rectangle}
For a protocol $\Protocol$ that takes the input from $\Lambda^n\times\Lambda^n$ and uses $c$ bits of deterministic communication, let $\Lambda^n\times\Lambda^n=\biguplus_{m\in\{0,1\}^c}R_m$ be a partition such that for each $m\in\{0,1\}^c$, the transcript of $\Protocol$ is $m$ on any input $(x,y)\in R_m$.
There exists a deterministic decision tree which queries at most $\frac{200c}b$ coordinates in $[n]$ such that: for any outcome of the query, there is a transcript $m\in\{0,1\}^c$ and a rectangle $R\subseteq R_m$ being dense on $(J,z)$, where $J\subseteq[n]$ is the set of unqueried coordinates and $z\in\{0,1\}^{\bar J}$ is the query answer.
\end{lemma}

To prove \Cref{lem:construct_rectangle}, the following results in \cite{CFKMP19} would be useful.
The proof is given in \Cref{sec:appendix} for completeness.

\begin{proposition}[{\cite[Proposition 11]{CFKMP19}}]
\label{prop:restore_density}
Let $J$ be a set of coordinates and $0<\delta\le1$.
For a random variable $X\in\Lambda^J$, let $I\subseteq J$ be any maximal subset such that $H_\infty(X_I)<\delta\cdot b\cdot|I|$.
Let $\alpha\in\Lambda^I$ be any value such that $\Dist{X_I}(\alpha)>2^{-\delta b|I|}$.
Then $X_{J\setminus I}\mid X_I=\alpha$ is $\delta$-dense.
\end{proposition}

\begin{definition}
\label{def:bad_value}
Let $J$ be a set of coordinates and $\alpha\in\Lambda^J$ and $0<\delta\le1$.
For a random variable $Y\in\Lambda^J$, $\alpha$ is $\delta$-bad for $Y$ if there is an $I\subseteq J$ and a $z\in\{0,1\}^I$ such that $Y_{J\setminus I}\mid G^I(\alpha,Y)=z$ is not $\delta$-dense, or $\Dist{G^I(\alpha,Y)}(z)<2^{-|I|-1}$.
\end{definition}

\begin{lemma}[{\cite[Theorem 7]{CFKMP19}}]
\label{lem:few_bad_value}
Let $J$ be a set of coordinates and $X,Y\in\Lambda^J$ be independent random variables which are $\delta_X$-dense and $\delta_Y$-dense, respectively.
If $\delta_X+\delta_Y\ge1.4$ and $\delta_Y\ge0.99$, the probability that $X$ takes a value that is $0.44$-bad for $Y$ is at most $1/n$.
\end{lemma}

\begin{proof}[Proof of \Cref{lem:construct_rectangle}]

Based on $\Protocol$, the desired decision tree is constructed as in \Cref{alg:construct_rectangle}.

\begin{algorithm}[ht]
\caption{A decision tree yielding dense rectangles}
\label{alg:construct_rectangle}
\begin{algorithmic}[1]
\STATE $U,V\gets\Lambda^n$ and $J\gets[n]$.
\STATE Although $U,V$ may vary, we always define $R=U\times V$ and random variables $X\sim U,Y\sim V$.
\STATE The initial value of $m\in\{0,1\}^c$ and $z\in\{0,1\}^n$ can be arbitrary.
\FOR{$i\in[c]$}

\IF{$i$ is odd}
\label{line:communication_start}
\STATE $U_a\gets\LL x\in U:m_i(x,m_1\cdots m_{i-1})=a\RL$ for $a\in\{0,1\}$.
\STATE $m_i\gets\Argmax_{a\in\{0,1\}}|U_a|$ and $U\gets U_{m_i}$.
\ELSE
\STATE $V_a\gets\LL y\in V:m_i(y,m_1\cdots m_{i-1})=a\RL$ for $a\in\{0,1\}$.
\STATE $m_i\gets\Argmax_{a\in\{0,1\}}|V_a|$ and $V\gets V_{m_i}$.
\ENDIF
\label{line:communication_end}
\WHILE{$X_J$ or $Y_J$ is not $0.99$-dense}

\STATE Assume that $X_J$ is not $0.99$-dense.
It is symmetric for $Y_J$ not being dense.
\STATE $U\gets\LL x\in U:x_J\textrm{ is not }0.44\textrm{-bad for }Y_J\RL$.
\label{line:remove_bad_value}
\STATE Set $I\subseteq J$ to be any maximal subset such that $H_\infty(X_I)<0.99b|I|$.
\STATE Set $\alpha\in\Lambda^I$ to be any value such that $\Dist{X_I}(\alpha)>2^{-0.99b|I|}$.
\STATE Make queries to $I$ and get $z_I\in\{0,1\}^I$.
\STATE $U\gets\LL x\in U:x_I=\alpha\RL$, $V\gets\LL y\in V:G^I(\alpha,y_I)=z_I\RL$, and $J\gets J\setminus I$.

\ENDWHILE

\ENDFOR
\end{algorithmic}
\end{algorithm}

We first prove that any set of unqueried coordinates $J\subseteq[n]$, query answer $z\in\{0,1\}^{\bar J}$, transcript $m$, and rectangle $R$ obtained by \Cref{alg:construct_rectangle} satisfy that $R\subseteq R_m$ is dense on $(J,z)$.
Then we bound the number of queries.
To prove the density of $R$, the following loop invariants are sufficient.

\begin{proposition}
At the start and end of each iteration of the for-loop in \Cref{alg:construct_rectangle}, it holds that $R$ is dense on $(J,z_{\bar J})$.
At the start and end of each iteration of the while-loop, it holds that either $X_J$ is $0.44$-dense and $Y_J$ is $0.99$-dense, or $X_J$ is $0.99$-dense and $Y_J$ is $0.44$-dense, where the random variables $X\sim U,Y\sim V$.
\end{proposition}

\begin{proof}
We prove this by induction.
At the start of the first iteration ($i=1$) of the for-loop, $R$ is dense on $(J,z_{\bar J})$ trivially since $R=\Lambda^n\times\Lambda^n$ and $J=[n]$.
We assume by the induction hypothesis that $R$ is dense on $(J,z_{\bar J})$ at the start of the $i$-th iteration of the for-loop.
We consider odd $i$ and the case where $i$ is even is symmetric.

Before executing \Cref{line:communication_start}, $G^{\bar J}(X,Y)\equiv z_{\bar J}$ and $X_J,Y_J$ are $0.99$-dense by the induction hypothesis, where $X\sim U,Y\sim V$.
After executing \Cref{line:communication_end}, $U$ is replaced by $U'$ with $|U'|\ge|U|/2$.
For each nonempty $I\subseteq J$,
$$H_\infty(X'_I)\ge H_\infty(X_I)-1\ge0.99b|I|-1\ge0.98b|I|,$$
where $X'\sim U'$.
So $X'_J$ is at least $0.98$-dense.

\begin{quote}
We now temporarily turn to proving the loop invariant for the while-loop.
Upon entering the while-loop for the first time, $X_J$ is $0.98$-dense (and, of-course, $0.44$-dense), and $Y_J$ is still $0.99$-dense, where $X\sim U,Y\sim V$.
The loop invariant holds for the base case.
We assume by the induction hypothesis that $X_J$ is $0.44$-dense and $Y_J$ is $0.99$-dense at the start of some iteration of the while-loop.
The case where $X_J$ is $0.99$-dense and $Y_J$ is $0.44$-dense is symmetric.

Before executing \Cref{line:remove_bad_value}, $X_J$ is $0.44$-dense and $Y_J$ is $0.99$-dense by  the induction hypothesis, where $X\sim U,Y\sim V$.
After executing \Cref{line:remove_bad_value}, $U$ is replaced by $U'$ with the bad values removed.
Then $X'$ is not $0.44$-bad for $Y$ where $X'\sim U'$.
Let $I\subseteq J$ be any maximal subset such that $H_\infty(X'_I)<0.99b|I|$, $\alpha\in\Lambda^I$ be any value such that $\Dist{X'_I}(\alpha)>2^{-0.99b|I|}$, and $z_I\in\{0,1\}^I$ be the answer of querying coordinates $I$.
Now $U$ is replaced by $U''$ such that $X''_I\equiv\alpha$, where $X''\sim U''$.
So $X''=\LS X'\mid X'_I=\alpha\RS$ and by \Cref{prop:restore_density},
$$X''_{J\setminus I}=\LS X'_{J\setminus I}\mid X'_I=\alpha\RS\textrm{ is }0.99\textrm{-dense}.$$
Then $V$ is replaced by $V'$ such that $G^I(\alpha,Y'_I)\equiv z_I$, where $Y'\sim V'$.
We have $Y'=\LS Y\mid G^I(\alpha,Y_I)=z_I\RS$.
As $X'_J$ is not $0.44$-bad for $Y_J$ and $\Dist{X'_I}(\alpha)>0$, by \Cref{def:bad_value},
$$Y'_{J\setminus I}=\LS Y_{J\setminus I}\mid G^I(\alpha,Y_I)=z_I\RS\textrm{ is }0.44\textrm{-dense}.$$
Besides, $G^I(X'',Y')\equiv G^I(\alpha,Y'_I)\equiv z_I$ and $I$ is removed from $J$.
So the loop invariant holds at the end of the current iteration of the while-loop, and will still hold at the start of the next iteration of the while-loop.
\end{quote}

At the end of the $i$-th iteration of the for-loop, the loop condition of the while-loop is violated.
So $X_J,Y_J$ are both $0.99$-dense, where $X\sim U,Y\sim V$.
And $G^{\bar J}(X,Y)\equiv z_{\bar J}$ since whenever some $I$ is removed from $J$, we have $G^I(X,Y)\equiv z_I$.
Therefore, $R$ is dense on $(J,z_{\bar J})$ at the end of the $i$-th iteration of the for-loop, also at the start of the $(i+1)$-th iteration.
\end{proof}

Let the transcript $m\in\{0,1\}^c$ and the rectangle $R\subseteq\Lambda^n\times\Lambda^n$ be obtained by \Cref{alg:construct_rectangle}.
We prove that $R\subseteq R_m$ by induction on $i\in[c]$.

For $i\in[c]$, let $\Lambda^n\times\Lambda^n=\biguplus_{m'\in\{0,1\}^i}R_{m'}$ be the partition such that for each $m'\in\{0,1\}^i$, the transcript of the first $i$ bits is $m'$ on any input $(x,y)\in R_{m'}$.
Assume by the induction hypothesis that $R\subseteq R_{m_1\cdots m_{i-1}}$at the start of the $i$-th iteration of the for-loop.
After executing \Cref{line:communication_end}, $R\subseteq R_{m_1\cdots m_i}$.
For the rest of the for-loop, the rectangle $R$ is just replaced by its sub-rectangle.
So $R\subseteq R_{m_1\cdots m_i}$ at the end of the $i$-th iteration of the for-loop, also at the start of the $(i+1)$-th iteration.
And finally, $R\subseteq R_m$.

\bigskip

To bound the number of queries, we use the following potential function as in \cite{CFKMP19}:
$$\Delta(U,V,J)=2b|J|-H_\infty(X_J)-H_\infty(Y_J)$$
for $U,V\subseteq\Lambda^n$ and $J\subseteq[n]$, where $X\sim U,Y\sim V$.
It holds that $\Delta(U,V,J)\ge0$.
In \Cref{alg:construct_rectangle}, the potential function $\Delta(U,V,J)$ initially equals $0$.
Intuitively, during the execution, the potential function increases by at most $O(1)$ for each bit transmitted, and decreases by at least $O(b|I|)$ for each subset $I$ removed from $J$.
\begin{itemize}
\item Whenever executing \Cref{line:communication_start} - \Cref{line:communication_end}, either $U$ is replaced by $U'$ with $|U'|\ge|U|/2$, or $V$ is replaced by $V'$ with $|V'|\ge|V|/2$.
For the first case, the potential function increases by $H_\infty(X_J)-H_\infty(X'_J)\le1$, where $X'\sim U'$.
Similarly, for the second case, the potential function increases by $H_\infty(Y_J)-H_\infty(Y'_J)\le1$, where $Y'\sim V'$.
\item For each iteration of the while-loop triggered because $X_J$ is not $0.99$-dense.
Firstly, $U$ is replaced by $U'$ with the bad values removed.
Since $X_J$ is at least $0.44$-dense and $Y_J$ is $0.99$-dense.
By \Cref{lem:few_bad_value}, $X_J$ is $0.44$-bad for $Y_J$ with probability at most $1-|U'|/|U|\le1/n$.
So $|U'|\ge|U|/2$, and the potential function increases by at most $1$.

Now $X'_J$ is not $0.44$-bad for $Y_J$, where $X'\sim U'$.
Let $I\subseteq J$ be any maximal subset such that $H_\infty(X'_I)<0.99b|I|$, and $\alpha\in\Lambda^I$ be any value such that $\Dist{X'_I}(\alpha)>2^{-0.99b|I|}$.
$U$ is replaced by $U''$ such that $X''_I\equiv\alpha$, where $X''\sim U''$.
So $|U''|\ge|U'|\cdot2^{-0.99b|I|}$, and the potential function increases by $H_\infty(X'_J)-H_\infty(X''_J)\le0.99b|I|$.

Then $V$ is replaced by $V'$ such that $G^I(\alpha,Y'_I)\equiv z_I$ for some $z_I\subseteq\{0,1\}^I$, where $Y'\sim V'$.
As $X'_J$ is not $0.44$-dense for $Y_J$ and $\Dist{X'_I}(\alpha)>0$,
by \Cref{def:bad_value}, $|V'|/|V|=\Dist{G^I(\alpha,Y_I)}(z_I)\ge2^{-|I|-1}$.
The potential function increases by $H_\infty(Y_J)-H_\infty(Y'_J)\le|I|+1$.

Finally, $I$ is removed from $J$.
The potential function increases by
$$-2b|I|+H_\infty(X''_J)-H_\infty(X''_{J\setminus I})+H_\infty(Y'_J)-H_\infty(Y'_{J\setminus I})\le-2b|I|+0+b|I|=-b|I|,$$
where $H_\infty(X''_J)=H_\infty(X''_{J\setminus I})$ because $X''_I\equiv\alpha$.

Overall, the potential function increases by at most $1+0.99b|I|+|I|+1-b|I|\le-0.005b|I|$.
For iterations triggered because $Y_J$ is not $0.99$-dense, the argument is symmetric.
\end{itemize}
At the end, the potential function is at most $c-0.005b|\bar J|\ge0$.
The number of queries $|\bar J|\le\frac{200c}b$.
\end{proof}
\section{Lifting Theorem for Hybrid Classical-Quantum Communication}

Let $n$ be sufficiently large.
Recall that the gadget $G:\Lambda\times\Lambda\to\{0,1\}$ is the inner product function over the Boolean domain, where $\Lambda=\{0,1\}^b$ and $b=20\log n$.
The following is a lifting-style statement: given a hybrid protocol that transmits $c$ classical bits followed by $q$ quantum bits and solves $f\circ G^n$, there is a $O\LS\frac cb\RS$-depth decision tree such that the outer function $f$ restricted to any outcome of the query has approximate degree $O\LS\frac qb\RS$.

\begin{theorem}
\label{thm:hybrid_lifting}
Let $\Protocol(c,q)$ be a hybrid protocol that allows $c$ bits of deterministic communication in advance, and uses $q$ qubits of quantum communication with prior entanglement (as defined in \Cref{sec:hybrid_communication}).
For any $f:\{0,1\}^n\to\{\pm1\}$, if $\Protocol$ computes $f\circ G^n$ with error $0.1$, there exists a deterministic decision tree which queries at most $\frac{200c}b$ coordinates in $[n]$ such that: for any outcome of the query, $\Deg_{1/3}(f_{J,z})=O\LS\frac qb\RS$, where $J\subseteq[n]$ is the set of unqueried coordinates and $z\in\{0,1\}^{\bar J}$ is the query answer.
\end{theorem}
\begin{proof}
Let $\Lambda^n\times\Lambda^n=\biguplus_{m\in\{0,1\}^c}R_m$ be a partition such that for each $m\in\{0,1\}^c$, $R_m$ contains all inputs on which the transcript of $\Protocol$ is $m$.
If $\Protocol$ computes $f\circ G^n$ with error $0.1$, $Q^*_{0.1}(F_m)\le q$ by \Cref{prop:qcc_on_rectangle}, where $F_m$ is the function $f\circ G^n$ with its input restricted to $R_m$.

By \Cref{lem:construct_rectangle}, there is a deterministic decision tree which queries at most $\frac{200c}b$ coordinates in $[n]$ such that for any set of unqueried coordinates $J\subseteq[n]$ and query answer $z\in\{0,1\}^{\bar J}$, we can find a transcript $m\in\{0,1\}^c$ and a rectangle $R\subseteq R_m$ which is dense on $(J,z)$.
Therefore,
\begin{itemize}
\item $G^{\bar J}(x,y)=z$ for all $(x,y)\in R$.
\item $X_J,Y_J$ are $0.99$-dense.
Define $F:U\times V\to\{\pm1\}$ such that
$$F(x,y)=f_{J,z}(G^J(x,y))=f\circ G^n(x,y),\qquad\forall\ (x,y)\in R.$$
By \Cref{thm:lifting_on_rectangle}, $Q^*_{0.1}(F)=\Omega\LS\Deg_{1/3}\LS f_{J,z}\RS\cdot b\RS$.
\end{itemize}
Since $R\subseteq R_m$, $Q^*_{0.1}(F)\le Q^*_{0.1}(F_m)\le q$.
So $\Deg_{1/3}\LS f_{J,z}\RS=O\LS\frac qb\RS$.
\end{proof}

As a corollary of \Cref{thm:hybrid_lifting}, we show the hardness of the classical-quantum trade-off.

\begin{theorem}
\label{thm:classical_quantum_tradeoff}
Let $\Protocol(c,q)$ be a hybrid protocol that allows $c$ bits of deterministic communication in advance, and uses $q$ qubits of quantum communication with prior entanglement.
For any function $f:\{0,1\}^n\to\{\pm1\}$, if $\Protocol$ computes $f\circ G^n$ with error $0.1$, we have
\begin{itemize}
\item $q=\Omega\LS\sqrt{\Deg(f)}\cdot b\RS$ if $c\le\frac{\Deg(f)\cdot b}{300}$.
\item $q=\Omega\LS\sqrt{\Bs(f)}\cdot b\RS$ if $c\le\frac{\Bs(f)\cdot b}{300}$.
\end{itemize}
\end{theorem}
\begin{proof}
We construct an input $z\in\{0,1\}^n$ of the decision tree in \Cref{thm:hybrid_lifting} by taking a walk down the decision tree.
The set of unqueried coordinates $J$ initially equals $[n]$.
For each queried coordinate $i\in[n]$, there is a $z_i\in\{0,1\}$ such that $\Deg\LS f_{J\setminus\{i\},z_{\bar J}z_i}\RS\ge\Deg\LS f_{J,z_{\bar J}}\RS-1$ as $\Deg(\cdot)$ is entropic by \Cref{prop:deg_is_entropic}.
We make a query to $i$ and let the query answer be such a $z_i$.
Then $i$ is removed from $J$.
At the end, $\Deg\LS f_{J,z_{\bar J}}\RS\ge\Deg(f)-|\bar J|$.
By \Cref{thm:hybrid_lifting},
\begin{itemize}
\item $|\bar J|\le\frac{200c}b$.
\item $\Deg_{1/3}(f_{J,z})=O\LS\frac qb\RS$.
\end{itemize}
So $q=\Omega\LS\Deg_{1/3}\LS f_{J,z_{\bar J}}\RS\cdot b\RS=\Omega\LS\sqrt{\Deg\LS f_{J,z_{\bar J}}\RS}\cdot b\RS$ \cite[Theorem 4]{ABKRT21}.
If $c\le\frac{\Deg(f)\cdot b}{300}$,
$$\Deg\LS f_{J,z_{\bar J}}\RS\ge\Deg(f)-|\bar J|\ge\Deg(f)-\frac{200c}b\ge\frac{\Deg(f)}3,$$
and $q=\Omega\LS\sqrt{\Deg(f)}\cdot b\RS$.

\bigskip

The lower bound with respect to block sensitivity $\Bs(\cdot)$ follows the same argument as $\Bs(\cdot)$ is entropic by \Cref{prop:bs_is_entropic} and $\Deg_{1/3}\LS f_{J,z}\RS=\Omega\LS\sqrt{\Bs\LS f_{J,z}\RS}\RS$ \cite[Lemma 3.8]{NS94}.
\end{proof}

As a corollary of \Cref{thm:classical_quantum_tradeoff}, we get a nearly tight bound for read-once formula $f$.
A read-once formula, which consists of AND gates, OR gates, and NOT gates, is a formula in which each variable appears exactly once.
We may let $f$ output $1$ if the corresponding formula outputs $0$, and $-1$ if the formula outputs $1$.

\begin{corollary}
\label{coro:for_readonce}
Let $\Protocol(c,q)$ be a hybrid protocol that allows $c$ bits of deterministic communication in advance, and uses $q$ qubits of quantum communication with prior entanglement.
Let $f$ be a read-once formula on $n$ bits.
if $\Protocol$ computes $f\circ G^n$ with error $0.1$ and $c\le\frac{nb}{300}$, then $q=\widetilde\Theta\LS\sqrt{n}\cdot b\RS$.
\end{corollary}
\begin{proof}
The degree of any read-once formula is $n$ \cite[Lemma 27]{ABKRT21}.
By \Cref{thm:classical_quantum_tradeoff}, when $c\le\frac{nb}{300}$, we obtain $q=\Omega\LS\sqrt n\cdot b\RS$.

On the other hand, the quantum query complexity of any read-once formula is $\Theta(\sqrt{n})$ \cite{CKK12}.
By the BCW simulation \cite{BCW98}, this yields a protocol with $q=O\LS\sqrt n\cdot b\log n\RS$.
\end{proof}

\bibliographystyle{alpha}
\bibliography{main}

\appendix

\section{Deferred Proofs}
\label{sec:appendix}

Let $n$ be sufficiently large.
Recall $G:\Lambda\times\Lambda\to\{0,1\}$ is the inner product function over the Boolean domain, where $\Lambda=\{0,1\}^b$ and $b=20\log n$.

\begin{proof}[Proof of \Cref{lem:pointwise_uniform}]
For any $z\in\{0,1\}^J$,
\begin{align*}
\Dist{G^J(X,Y)}(z) & =\sum_{\alpha,\beta\in\Lambda^J}\Dist X(\alpha)\Dist Y(\beta)\prod_{i\in J}\frac{1+(-1)^{z_i}(-1)^{\IP{\alpha_i}{\beta_i}}}2 \\
& =2^{-|J|}\sum_{\alpha,\beta\in\Lambda^J}\Dist X(\alpha)\Dist Y(\beta)\sum_{I\subseteq J}\Chi I(z)(-1)^{\IP{\alpha_I}{\beta_I}} \\
& =2^{-|J|}\sum_{I\subseteq J}\Chi I(z)\sum_{\alpha,\beta\in\Lambda^I}\Dist{X_I}(\alpha)\Dist{Y_I}(\beta)(-1)^{\IP\alpha\beta} \\
& =2^{-|J|}\LS1+\sum_{\emptyset\ne I\subseteq J}\Chi I(z)\sum_{\alpha,\beta\in\Lambda^I}\Dist{X_I}(\alpha)\Dist{Y_I}(\beta)(-1)^{\IP\alpha\beta}\RS.
\end{align*}
Let $H$ be the $2\times2$ Hadamard matrix.
\begin{align*}
\Abs{\sum_{\emptyset\ne I\subseteq J}\Chi I(z)\sum_{\alpha,\beta\in\Lambda^I}\Dist{X_I}(\alpha)\Dist{Y_I}(\beta)(-1)^{\IP\alpha\beta}} & \le\sum_{\emptyset\ne I\subseteq J}\Abs{\sum_{\alpha,\beta\in\Lambda^I}\Dist{X_I}(\alpha)\Dist{Y_I}(\beta)(-1)^{\IP\alpha\beta}} \\
& \le\sum_{\emptyset\ne I\subseteq J}\Abs{(\Dist{X_I})^\Trans H^{\otimes b|I|}(\Dist{Y_I})} \\
& \le\sum_{\emptyset\ne I\subseteq J}\Norm{\Dist{X_I}}\cdot\Norm{H^{\otimes b|I|}}\cdot\Norm{\Dist{Y_I}} \\
& \le\sum_{\emptyset\ne I\subseteq J}2^{-\delta_Xb|I|/2}\cdot2^{b|I|/2}\cdot2^{-\delta_Yb|I|/2} \\
& \le\sum_{\emptyset\ne I\subseteq J}2^{-\LS\delta_X+\delta_Y-1\RS b|I|/2} \\
& \le\sum_{\emptyset\ne I\subseteq J}n^{-4|I|}\le n^{-2}.
\qedhere
\end{align*}
\end{proof}

\begin{proof}[Proof of \Cref{prop:restore_density}]
Assume for the sake of contradiction that $X_{J\setminus I}\mid X_I=\alpha$ is not $\delta$-dense, there is a nonempty $K\subseteq J\setminus I$ and a $\beta\in\Lambda^K$ such that $\Dist{X_K\mid X_I=\alpha}(\beta)>2^{-\delta b|K|}$.
This implies that
$$\Dist{X_{I\uplus K}}(\alpha,\beta)=\Dist{X_I}(\alpha)\cdot\Dist{X_K\mid X_I=\alpha}(\beta)>2^{-\delta b\LS|I|+|K|\RS}.$$
It contradicts the maximality of $I$.
\end{proof}

\begin{proof}[Proof of \Cref{lem:few_bad_value}]
The proof is as in \cite[Section 5]{CFKMP19}.
For $I\subseteq J$, $K\subseteq J\setminus I$, and $\beta\in\Lambda^K$, let $\phi_{I,K,\beta}:\Lambda^I\to\Real$ denote the function that maps $\gamma\in\Lambda^I$ to $\Dist{Y_{I\uplus K}}(\gamma,\beta)$.
So $\phi_I=\Dist{Y_I}$ for $K=\emptyset$.

\begin{definition}
\label{def:biased_value}
A value $\alpha\in\Lambda^J$ is $\eta$-biased for $Y$ with respect to $K\subseteq J$ if for every $\beta\in\Lambda^K$ and nonempty $I\subseteq J\setminus K$, $\Abs{\widehat\phi_{I,K,\beta}(\alpha_I)}\le\eta\cdot2^{-1.1b|I|}$.
\end{definition}

\begin{lemma}
\label{lem:few_non_biased_value}
The probability that $X$ is not $\eta$-biased for $Y$ w.r.t. $K\subseteq J$ is at most $\frac{2^{-\delta_Yb|K|}}{\eta^2}\cdot n^{-2}$.
\end{lemma}
\begin{proof}
For every $I\subseteq J\setminus K$, $\alpha\in\Lambda^J$ is not $\eta$-biased due to $\alpha_I$ if there is a $\beta\in\Lambda^K$ such that $\Abs{\widehat\phi_{I,K,\beta}(\alpha_I)}>\eta\cdot2^{-1.1b|I|}$, and then $\sum_{\beta\in\Lambda^K}\widehat\phi_{I,K,\beta}(\alpha_I)^2>\eta^2\cdot 2^{-2.2b|I|}$.
Note that
$$\sum_{\alpha_I\in\Lambda^I,\beta\in\Lambda^K}\widehat\phi_{I,K,\beta}(\alpha_I)^2=2^{-b|I|}\sum_{\gamma\in\Lambda^{I\uplus K}}\Dist{Y_{I\uplus K}}(\gamma)^2\le2^{-b|I|}\cdot2^{-\delta_Yb\LS|I|+|K|\RS}.$$
The number of $\alpha_I$ that causes $\alpha$ to be not $\eta$-biased is at most $\frac{2^{-b|I|}\cdot2^{-\delta_Yb\LS|I|+|K|\RS}}{\eta^2\cdot 2^{-2.2b|I|}}$.
Since $X$ is $\delta_X$-dense, For any $\alpha_I$, $\Pr[X_I=\alpha_I]\le2^{-\delta_Xb|I|}$.
Hence, the probability that $X$ is not $\eta$-biased for $Y$ with respect to $K$ due to $\alpha_I$ is at most $\frac{2^{-(\delta_X+\delta_Y-1.2)b|I|-\delta_Yb|K|}}{\eta^2}\le\frac{2^{-0.2b|I|-\delta_Yb|K|}}{\eta^2}$.
By the union bound, the probability that $X$ is not $\eta$-biased for $Y$ with respect to $K$ is at most
\begin{equation*}
\sum_{\emptyset\ne I\subseteq J\setminus K}\frac{2^{-0.2b|I|-\delta_Yb|K|}}{\eta^2}=\frac{2^{-\delta_Yb|K|}}{\eta^2}\sum_{\emptyset\ne I\subseteq J\setminus K}2^{-0.2b|I|}\le\frac{2^{-\delta_Yb|K|}}{\eta^2}\cdot n^{-2}.
\qedhere
\end{equation*}
\end{proof}

\begin{lemma}
\label{lem:bias_is_not_bad}
If $\alpha\in\Lambda^J$ is $\frac12$-biased for $Y$ with respect to $K=\emptyset$, and is $2^{-\frac{\delta_Y}{2.2}b|K|}$-biased for $Y$ with respect to each nonempty $K\subseteq J$, it holds that $\alpha$ is not $0.44$-bad for $Y$.
\end{lemma}
\begin{proof}
We arbitrarily fix $I\subseteq J$ and $z\in\{0,1\}^I$.
For $K=\emptyset$,
\begin{align*}
\Dist{G^I(\alpha_I,Y_I)}(z) & =2^{-|I|}\LS1+\sum_{\emptyset\ne L\subseteq I}\Chi L(z)\sum_{\gamma\in\Lambda^L}\Dist{Y_L}(\gamma)(-1)^{\IP{\alpha_L}{\gamma}}\RS \\
& =2^{-|I|}\LS1+\sum_{\emptyset\ne L\subseteq I}\Chi L(z)\cdot2^{b|L|}\cdot\hatDist{Y_L}(\alpha_L)\RS \\
% & \ge2^{-|I|}\LS1-\sum_{\emptyset\ne L\subseteq I}2^{b|L|}\Abs{\hatDist{Y_L}(\alpha_L)}\RS \\
& \ge2^{-|I|}\LS1-\sum_{\emptyset\ne L\subseteq I}2^{b|L|}\cdot\frac12\cdot2^{-1.1b|I|}\RS \\
& =2^{-|I|}\LS1-\frac12\sum_{\emptyset\ne L\subseteq I}2^{-0.1b|I|}\RS \\[1ex]
& \ge2^{-|I|-1}.
\end{align*}
Similarly for nonempty $K\subseteq J\setminus I$ and $\beta\in\Lambda^K$,
\begin{align*}
\Dist{\LS Y_K,G^I(\alpha_I,Y_I)\RS}(\beta,z) & \le2^{-|I|}\LS\Dist{Y_K}(\beta)+2^{-\frac{\delta_Y}{2.2}b|K|}\RS \\
& \le2^{-|I|}\LS2^{-\delta_Yb|K|}+2^{-\frac{\delta_Y}{2.2}b|K|}\RS \\
& \le2^{-|I|}\cdot2\cdot2^{-\frac{\delta_Y}{2.2}b|K|} \\[1ex]
& =2^{-|I|+1-\frac{\delta_Y}{2.2}b|K|}.
\end{align*}
Therefore,
\begin{align*}
\Dist{Y_K\mid G^I(\alpha_I,Y_I)=z}(\beta) & =\frac{\Dist{\LS Y_K,G^I(\alpha_I,Y_I)\RS}(\beta,z)}{\Dist{G^I(\alpha_I,Y_I)}(z)} \\
& \le\frac{2^{-|I|+1-\frac{\delta_Y}{2.2}b|K|}}{2^{-|I|-1}} \\
& \le2^{-\frac{\delta_Y}{2.2}b|K|+2} \\[1ex]
& \le2^{-0.44b|K|}
\end{align*}
for any nonempty $K\subseteq J\setminus I$ and $\beta\in\Lambda^K$.
That is, $Y_{J\setminus I}\mid G^I(\alpha_I,Y_I)=z$ is $0.44$-dense.
\end{proof}

Combining \Cref{lem:few_non_biased_value} and \Cref{lem:bias_is_not_bad}, and a union bound over $K\subseteq J$, the probability that $X$ is $0.44$-bad for $Y$ is at most
\begin{align*}
4\cdot n^{-2}+\sum_{\emptyset\ne K\subseteq J}\frac{2^{-\delta_Yb|K|}}{2^{-\frac{\delta_Y}{1.1}b|K|}}\cdot n^{-2} & \le4\cdot n^{-2}+\sum_{\emptyset\ne K\subseteq J}2^{-0.09\delta_Yb|K|}\cdot n^{-2}\le\frac1n.
\qedhere
\end{align*}
\end{proof}

\end{document}